\documentclass[sigconf]{acmart}

\setcopyright{none} 
\usepackage{amsmath,amsfonts,amsthm}
\usepackage{mathtools}
\usepackage{algorithmic}
\usepackage{graphicx}
\usepackage{textcomp}
\usepackage{xcolor}
\usepackage{url}
\usepackage{xspace}
\usepackage{pgfplots}
\usetikzlibrary{pgfplots.groupplots}
\pgfplotsset{compat=1.3}
\usepackage[ruled,vlined, linesnumbered]{algorithm2e}
\def\BibTeX{{\rm B\kern-.05em{\sc i\kern-.025em b}\kern-.08em
    T\kern-.1667em\lower.7ex\hbox{E}\kern-.125emX}}

\usepackage{subcaption}

\usepackage{tikz}
\usepackage{multirow}
\usepackage{graphicx}
\usepackage{thmtools, thm-restate}

\usepackage{chemformula}

\theoremstyle{plain}
\newcommand\numberthis{\addtocounter{equation}{1}\tag{\theequation}}

\DeclarePairedDelimiterXPP\Expn[2]{\mathbb{E}_{#1}}{[}{]}{}{

#2
}
\newcommand{\Exp}{\mathrm{E}}
\newcommand{\Var}{\mathrm{Var}}
\newcommand{\Cov}{\mathrm{Cov}}
\DeclarePairedDelimiter\abs{\lvert}{\rvert}

\mathchardef\mhyphen="2D
\newcommand{\dasgupta}{\mathcal{C}_{\mathcal{D}}}
\newcommand{\dasguptabar}{\overline{\mathcal{C}}_{\mathcal{D}}}
\newcommand{\mle}{\mathcal{C}_{\mathcal{M}}}
\newcommand{\globalalgo}{{\sc CMN-DP}\xspace}
\newcommand{\localalgo}{{\sc GenTree}\xspace}

\newcommand{\dplocalalgo}{{\sc PrivaCT}\xspace}

\newcommand{\numnodes}{n}
\newcommand{\numstates}{N}

\newcommand{\baseline}{{\sc CMN}\xspace}

\newcommand{\tree}{T}
\newcommand{\lst}{L}
\newcommand{\rst}{R}

\newcommand{\Internal}{\Lambda}
\newcommand{\inode}{\lambda}
\newcommand{\leaves}[1]{\mathrm{leaves}({#1})}
\newcommand{\lca}{T[x\lor y]}

\newcommand{\ancestor}{ancestor}
\newcommand{\ancestors}{ancestors}

\newcommand{\opttree}{T^{OPT}}
\newcommand{\optcost}{OPT}
\newcommand{\alltrees}{\mathbb{T}}

\newcommand{\minoptcost}{\dasgupta(T_{clique})}
\newcommand{\sumweights}[2]{\sum S(#1, #2)}

\newcommand{\dasguptafull}[2]{\sum_{ij}{#1}(v_i, v_j)\cdot\leaves{#2[v_i\lor v_j]}}
\newcommand{\actcost}[1]{f({#1})}
\newcommand{\dpcost}[1]{g({#1})}
\newcommand{\nrv}{R}
\newcommand{\nrvi}[2]{R_{#1}^{#2}}
\newcommand{\numL}[2]{L_{#1,#2}}
\newcommand{\numbins}{K}
\newcommand{\dv}[1]{dv_{#1}}
\newcommand{\ndv}[1]{\bar{dv}_{#1}}
\newcommand{\dc}[2]{c_{#1}^{#2}}
\newcommand{\ndc}[2]{\bar{c}_{#1}^{#2}}
\newcommand{\idealprop}{ideal clustering property\xspace}
\newcommand{\fcostinode}[1]{\theta({#1})}
\newcommand{\costinode}{\theta}

\newcommand{\simscore}[2]{S({#1},{#2})}
\newcommand{\dpsimscore}[2]{\bar{S}({#1},{#2})}
\newcommand{\simmatrix}{S}
\newcommand{\dpsimmatrix}{\bar{S}}
\newcommand{\normL}{$L_1$-norm\xspace}

\newcommand{\custombox}[2]{\begin{center}\fbox{\parbox{3.3in}{\textit{Result
{#1}: {#2}}}\xspace}\end{center}}

\newcommand{\newtext}[1]{\textcolor{black}{#1}}

\newcommand{\lastfm}{\textit{lastfm}\xspace}
\newcommand{\douban}{\textit{douban}\xspace}
\newcommand{\delicious}{\textit{delicious}\xspace}

\newcommand{\itemAvg}{\textit{itemAvg}\xspace}
\newcommand{\friendsCF}{\textit{friendsCF}\xspace}
\newcommand{\privaCTCF}{\privact-CF\xspace}
\newcommand{\serec}{SERec\xspace}
\newcommand{\privact}{{\sc PrivaCT}\xspace}
\newcommand{\tool}{{\sc PrivaCT}\xspace}

\definecolor{amethyst}{rgb}{0.6, 0.4, 0.8}
\definecolor{azure(colorwheel)}{rgb}{0.0, 0.5, 1.0}
\definecolor{chocolate(traditional)}{rgb}{0.48, 0.25, 0.0}
\definecolor{onyx}{rgb}{0.06, 0.06, 0.06}

\begin{document}

\title{Private Hierarchical Clustering in Federated Networks}

\author{Aashish Kolluri}
\affiliation{%
 \institution{National University Of Singapore}
 \country{Singapore}
}

\author{Teodora Baluta}
\affiliation{%
 \institution{National University Of Singapore}
 \country{Singapore}
}

\author{Prateek Saxena}
\affiliation{%
 \institution{National University Of Singapore}
 \country{Singapore}
}

\begin{abstract}
Analyzing structural properties of social networks, such as identifying their clusters or finding their most central nodes, has many applications. However, these applications are not supported by federated social networks that allow users to store their social links locally on their end devices. In the federated regime, users want access to personalized services while also keeping their social links private. In this paper, we take a step towards enabling analytics on federated networks
with differential privacy guarantees about protecting the user links or contacts in the network. Specifically, we present the first work to compute hierarchical cluster trees using local differential privacy. Our algorithms for computing them are novel and come with theoretical bounds on the quality of the trees learned. The private hierarchical cluster trees enable a service provider to query the community structure around a user at various granularities without the users having to share their raw contacts with the provider. We demonstrate the utility of such queries by redesigning the state-of-the-art social recommendation algorithms for the federated setup. Our recommendation algorithms significantly outperform the baselines which do not use social contacts and are on par with the non-private algorithms that use contacts.
\end{abstract}

\maketitle

\section{Introduction}
\label{sec:intro}

Millions of users are moving towards more decentralized or federated services due to trust and privacy concerns
of centralized data storage~\cite{twitterMovingAway,millions-left-whatsapp}. Federated social
networks are a popular alternative to the centralized \newtext{networks} as the user's social connections are kept
private and stored on user-controlled devices. For instance, Signal, a federated social network, has
recently seen tens of millions of users joining the platform after WhatsApp announced a
controversial privacy policy update~\cite{whatsappPolicyUpdate}. This movement towards
decentralization has incentivized companies to develop techniques to port conventional end
applications to the federated setup~\cite{floc-whitepaper,google-federated-analytics,brave-ad}.

Conventional end applications for social and communication networks, such as personalized
recommendations and online advertising, require finding similar users on the network that have an
influence over a target user~\cite{jamali2009trustwalker,yan2009much}.  For instance, a user's
network neighbors and other users in their close community are known to influence the user's
behavior~\cite{golbeck2005computing}. Therefore, the ability to probe the close community of a
target user is important to analyze. In the centralized setup, such queries are trivial because
the users' data is stored in the servers of a centralized service provider and thus the whole network is available.
However, in the federated setup this network structure is not available to the service
provider. Moreover, the users of a service often derive value and benefits from
conventional end applications, while expecting a reasonable privacy guarantee on their sensitive data~\cite{brave-ad}.  This new paradigm raises the problem of redesigning conventional applications for the
federated setup, which in turn requires supporting queries such as ``Who are the users in the close community of a target user?''

It is useful to think of a federated social network as a graph. Users form vertices of the graph and the outgoing links of a user are only kept locally with it.
A starting point for answering queries over a federated network which preserve {\em privacy of the individual links} is the local differential privacy (LDP) framework~\cite{kasi2008FOCS}. The LDP regime
eliminates trust in a centralized authority and thus naturally fits the federated setup~\cite{diaspora,mastodon,peepeth}. In this setup, each user locally
adds noise to the query outputs computed on its raw data to preserve privacy before releasing it. Users
communicate with an untrusted authority which combines the noisy outputs of the users, possibly with
more than one round of communication, to compute the final result. 
In the LDP model, the noise added by every user in
multiple rounds of communication adversely affects the utility of a query. Consequently, queries
in the LDP regime have largely been restricted to simple statistical queries such as counts and
histograms over categorical, set-valued data, or key-value
pairs~\cite{hsu2012LDPHeavyHitters,duchi2023MinimaxRates,rappor2014CCS,ye2019KeyValueLDP,gu2019LDPRange},
with only a few exceptions~\cite{qin2017CCS}.
So, more complex queries that probe the sensitive link or community structure around a target user remain a challenge in the LDP regime.

As our first contribution, we propose learning a well-known data structure called a {\em hierarchical cluster tree} (HCT) over a federated network of users in the LDP regime. 
An HCT is a tree representation of the network which clusters similar vertices at
various levels.  At lower levels,  small tightly-connected clusters emerge; and, at higher levels, larger clusters and structural hubs are captured.
Observe that an HCT naturally captures the communities around a user at various granularities.
In the centralized setup, HCTs are commonly used in recommender
systems~\cite{shepitsen2008personalized,pinterestEngineering,bateniGoogle,bateni2017NIPS}, intrusion detection~\cite{khan2007VLDB}, link
prediction~\cite{lu2011Physica} and even in applications beyond computer
science~\cite{mostavi2008GB,corpet1988NAR,cole1996RAS,david1999Ties}.
Learning HCTs with LDP guarantees unlocks such applications in the federated setup.
The only existing approach to privately learn HCTs is for the centralized setup, where the central authority is entrusted with the whole network.
However, this approach incurs large amounts of noise if it were directly adapted to fit the LDP regime as it requires multiple queries on the network structure.
Therefore, we ask: \textit{Can we compute
HCTs over federated networks (graphs) with acceptable privacy and utility? }

We present the first algorithm called \dplocalalgo to learn HCTs over federated graphs in the LDP framework. In the process, we also design a novel randomized algorithm called \localalgo, to compute HCTs in the federated setup without differential privacy. To achieve this, we make the key observation that a recently proposed cost function to measure the quality of HCTs couples well
with a known differentially private construct called degree vectors. Our design choices are principled and guided throughout by theoretical utility analysis. Specifically, we show that \localalgo creates HCTs within $O(\frac{\log n}{n^2})$ approximation error of the ideal HCT in expectation, where $n$ is the number of vertices in the given federated graph. Its differential private version called \dplocalalgo has an additive approximation error term bounded by a quantity that depends only on $n$, not on the edge structure of the graph. Therefore, once we fix a graph with $n$ vertices, the expected error (or loss in utility) can be analytically bounded for various choices of the privacy budget $\epsilon$~\cite{privacybook}. Further, \dplocalalgo requires querying each user just once.

Finally, we show a concrete application that directly benefits from our private HCTs: social recommendation systems.
To alleviate issues such as lack of data when a new user joins the \newtext{centralized recommendation service}, or to increase the recommendation quality, social recommender systems use additional cross-site data such as users' social network~\cite{deng2013personalized,perera2017exploring,roy2012socialtransfer,osborne2012bieber}.
The key insight is to use the private HCTs to identify users in the close communities of a target
user. We show that \dplocalalgo can be readily integrated into both
traditional~\cite{tang2013social} and state-of-the-art~\cite{wang2018collaborative} algorithms to
provide recommendations in a setup where users' privacy is important and users do not rely on a
trusted central entity.
On $3$ commonly used social recommendation datasets, we demonstrate that \dplocalalgo-based private social recommendation algorithms significantly outperform the algorithms that do not use social information. At the same time, their performance is comparable to the non-private social recommendation algorithms. Therefore, \dplocalalgo can be used to improve recommendation quality by utilizing the users' social information from federated networks, while providing a strong privacy ($\epsilon=1$). 

To summarize, we claim the following contributions:
\begin{itemize}
    \item \textbf{\em Conceptual}: We propose the first work that learns hierarchical cluster trees from a federated network in the local differential privacy model, to the best of our knowledge. 
    \item \textbf{\em Technical}: We propose two novel algorithms (Section~\ref{sec:solution}) with utility guarantees: 1)
      \localalgo learns hierarchical cluster trees which are close to the ideal in expectation  and
      2) \dplocalalgo learns private hierarchical cluster trees whose utility loss can be bounded
      for various choices of privacy budget. We evaluate the quality of our learned cluster trees by
      comparing with the state-of-the-art solution developed for differential privacy in \newtext{a} centralized
      setup, showing empirical improvements in the obtained utility of the hierarchical cluster
      trees (Section~\ref{sec:eval}). Further, the observed difference in  empirical utility of our differentially-private and \newtext{non-private} versions of our algorithms is small.
    \item \textbf{\em End application}: Using \dplocalalgo-based HCTs,
  we redesign the state-of-the-art social recommendation algorithms such that users in the federated
  setup can be served \newtext{high-}quality recommendations without having to share their contacts with the
  service provider. Our recommendation algorithms outperform the baselines which do not use social connections and are on-par with the non-private baselines on $3$ popular social recommendation datasets (Section~\ref{sec:case-study}).
\end{itemize}

\section{Motivation \& Problem}
\label{sec:problem}

We illustrate our problem setup in the context of networks where edges are private information such as communication and social networks. The structure of social networks captures influence relationships between users which are instrumental in conventional applications such as serving recommendations~\cite{golbeck2005computing} and link prediction~\cite{lu2011link}. At the heart of these applications lies the idea that users' preferences are influenced by other users on the network.  A user's influence is projected beyond his immediate neighbors into the broader community and decreases as the distance on the network increases. Ideally, having the network structure allows us to understand a user's preferences based on his close community consisting of his neighbors, second-degree neighbors, and so on. Therefore, the queries that we are specifically interested in answering are: {\em ``Which users are in the close community of a target user?''}  where closeness dictates the granularity of exploration~\cite{osborne2012bieber,deng2013personalized}.
We show how these queries are useful in social recommendations. 

\subsection{Motivating Example}
Recommender systems are widely used to help users find interesting products,
services, or other users from a large set of choices. These systems are crucial
for the user experience and user retention~\cite{netflixblog}.
Despite their success, recommender systems suffer from the problem of data
sparsity. For example, when a new user joins the \newtext{recommendation system} there is no history to base
the recommendations on (known as the cold start problem).
To mitigate this problem, recommender systems benefit from cross-site
data such as online social networks~\cite{perera2017exploring,liu2019discrete,tang2013social}.
For example, YouTube recommender accuracy was increased using information from
Google+~\cite{deng2013personalized}, Twitter and Amazon were increased with
Cheetah Mobile~\cite{osborne2012bieber,perera2020towards}.
In fact, this has lead to an area of research called {\em social
recommendation}~\cite{golbeck2006generating} which is based on the idea that a
user's preferences are similar to or influenced by their friends and other
users in the close community.  The reason behind this observation can be
explained by social correlation theories~\cite{singla2008yes}.

\begin{figure}
    \centering
    \includegraphics[scale=0.23]{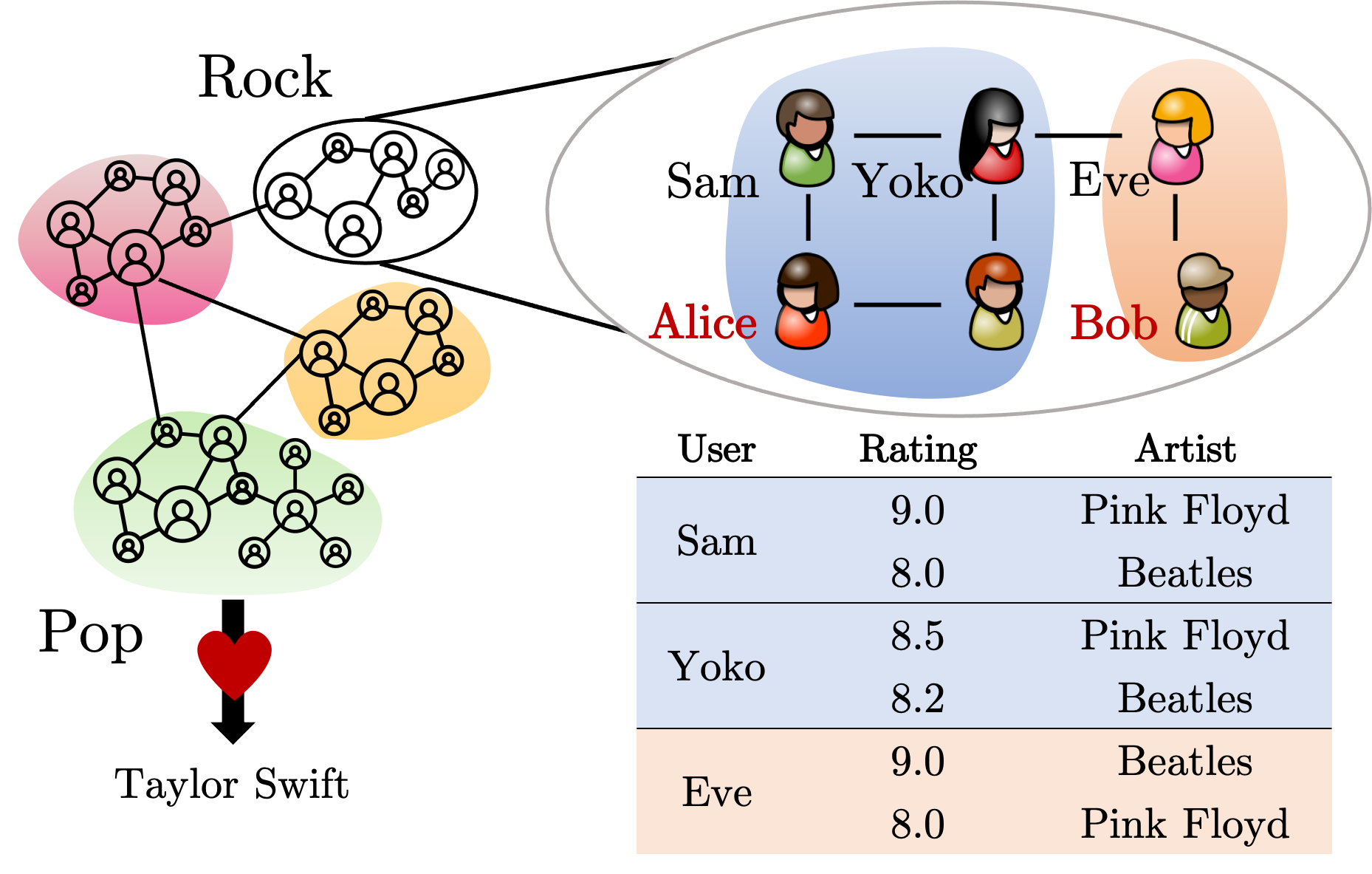}
    \caption{To recommend Alice and Bob an artist, we use the preferences of users in their close
      communities. Therefore, Alice is recommended Pink Floyd, then Beatles, and vice-versa for Bob.
      The amount of ``closeness'' can be set to just first-degree neighbors or to a higher degree
      depending on the scenario. Both of them are not recommended Taylor Swift, a pop artist popular
      in another community.
    }
    \label{fig:motivating_example}
\end{figure}

Consider a music streaming \newtext{service} that recommends artists to users.
On the streaming \newtext{service}, each user rates artists from $1$ to $10$ either
explicitly or implicitly by the number of times they listened to that particular
artist.
\newtext{Hence, the service stores user's item preferences on its centralized servers.}
The recommendation algorithm additionally uses social
contacts to make recommendations.
\newtext{The social contacts can be obtained by asking users to share their
contacts directly while signing up (e.g., Google or Facebook contacts).}
In Figure~\ref{fig:motivating_example}, we illustrate such a \newtext{service} where we
zoom into the community of users interested in rock music.
When new users Alice and Bob sign up for the streaming service, they connect with
their friends who belong to this broad community.
A good social recommendation algorithm recommends Alice the artists Pink Floyd
and the Beatles, in that order, and the same artists for Bob, but in reverse
order.
The intuition for this is captured by the community structure around these
users. Observe that Alice's close community consisting of her first- and
second-degree neighbors has rated Pink Floyd higher than the Beatles. For
example, Sam, her direct neighbor, has rated Pink Floyd $9.0$ and the Beatles,
$8.0$.
In the real world, this means that Alice's friends and friends of her friends
have rated Pink Floyd higher than the Beatles.
Alice's next closest community consisting of third-degree neighbors (Eve) rates
the Beatles higher than Pink Floyd but they still listen to and rate Pink Floyd
highly $>8.0$. However, Taylor Swift, a popular artist in another part of the
network, is rated lower in this community. Thus, Taylor Swift does not appear in
the top recommendations for either Alice or Bob.

Currently, centralized streaming services query for close communities around a target user.
In this work, we take a step
towards supporting such queries in
the federated setup.
This would enable conventional end applications in the federated setup without the need for users to share their contacts.

\subsection{Problem Setup}
\label{sec:prob_setup}
A federated social network is a graph $G:(V,E)$ with $n=|V|$ vertices.
Each vertex in $V$ corresponds (say) to a separate user who stores the list of its neighbors locally.
We assume that the graph is static thus there is no change in the set of
vertices or edges.  There is an untrusted central {\em authority} who knows the
registered users on the network. However, the central authority does not know
the edges of any user. Users want to keep their edges private.  This setup is
common in the real-world federated social networks~\cite{diaspora,mastodon}.
We assume that the users are honest and they want good services from the central
authority.  Therefore, the central authority can query the users on their
private edge information and they reply as long as they are given a {\em
reasonable} privacy guarantee.

In this setup, our goal is to enable the aforementioned queries without the central authority
knowing the edges of a user. At the same time, the user should not incur a prohibitive computation
or communication cost. We turn to a local differential privacy framework as it is widely regarded as
a strong privacy guarantee that can be given to a user, while allowing the central authority to
extract useful information from the user. 

\paragraph{\newtext{Differential Privacy}}
The differential privacy (DP) framework helps to bound the privacy loss incurred by a user by participating in a computation.

\begin{definition}[Differential Privacy]
For any two datasets, $D$ and $D'$ such that ($|D-D'|\leq1$), a randomized query $M:D\rightarrow{}
S$ satisfies $\epsilon$-DP if
\begin{align*}
Pr(M(D)\in s)\leq e^\epsilon Pr(M(D')\in s)
\end{align*}
Here, $s\subseteq S$ is any possible subset of all outputs of $M$.
\end{definition}

The $\epsilon$ parameter is called the privacy loss and it bounds the ratio of probabilities of observing any chosen output with any two input datasets $D$ and $D'$ that differ in one element.
The parameter $\epsilon$ captures the maximum loss in privacy resulting from
one run of the algorithm $M$. We can see that the lower the value of $\epsilon$, the closer the two output probability distributions are, therefore, the higher the privacy. Every time $M$ is run and its outputs are released the value of $\epsilon$ increases implying more loss of privacy.
However, once the outputs of $M$ are public then they can be reused many times without losing additional privacy as given by the {\em post-processing property}~\cite{privacybook}.
In our setup, the central authority is untrusted which corresponds to the \newtext{local differential privacy (LDP)} regime where each user stores the edge information locally and uses a DP mechanism on it.
Therefore, the input dataset for the $i$-th vertex in our problem is a binary array $D_i:\{0,1\}^{n}$ with $1$ in the $j$-th index if an edge exists between vertices $i$ and $j$, and $0$ otherwise. 
Since users do not release their edge information but the identities are known to the central
authority, we work with {\em edge local differential privacy}~\cite{qin2017CCS}. In this framework, the privacy of edges (each bit in the $D_i$) of a user ($v_i$) is preserved.

\paragraph{Enabling Queries in LDP}
Our LDP setup poses a significant challenge to design queries that explore the community around a user.
For instance, even a simple query such as ``reporting the immediate contacts'' is privacy-sensitive and
requires a user to add significant amounts of noise. To illustrate, the non-private answer for a
user $v_i$ would be the original bit vector $D_i$. To preserve LDP, we can use a popular strategy called
randomized response for this query. In randomized response, for each bit in $D_i$ the user reports the
original value of $D_i[j]$ with a high probability, say $70\%$, and reports the flipped value
otherwise. This strategy confers plausible deniability: even if the user $u_i$ reported that $u_j$ is their immediate contact, there is a $30\%$ chance that there is in fact no edge between $u_i$ and $u_j$. Observe that by using this strategy a user will claim to have $700,000$
contacts for a network of size $1,000,000$ even if they have just $70$ contacts.
Consequently, most
of the existing work provide LDP queries for less sensitive numerical queries such as ``number of
contacts'' or aggregate statistics.
Our goal is to support more complex queries that ask for the community around a user at various levels of closeness (i.e., second-degree, third-degree, and so on). If the private algorithm requires users to release their data for every degree of closeness, additional noise is required.
Therefore, exploring the community around users at different
granularity on a federated graph without adding significant noise is challenging.

 
\section{Hierarchical Cluster Trees}
\label{sec:hct}

\begin{figure}[t]
  \begin{center}
  \resizebox{0.9\linewidth}{!}{
    \includegraphics[scale=0.1]{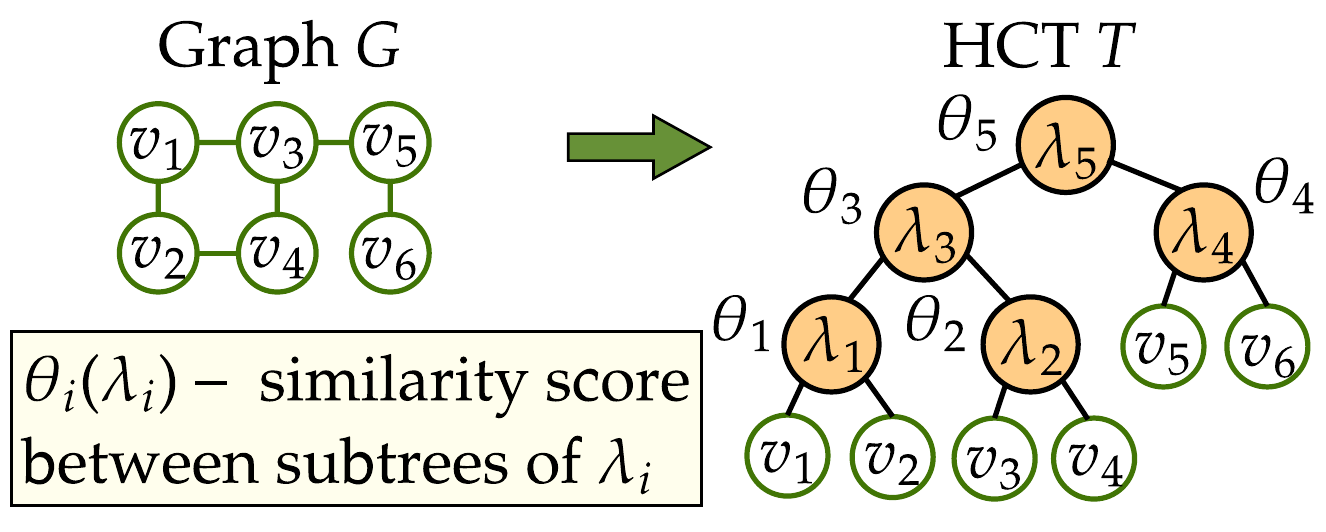}
  }
     \end{center}
     \caption{The network from Figure~\ref{fig:motivating_example} and its candidate HCT.}
  \label{fig:dend}
\end{figure}

To support complex community queries at different granularities without additional noise, we propose to use a well-known data structure called hierarchical cluster tree (HCT). HCTs are used in many applications ~\cite{shepitsen2008personalized, khan2007VLDB,lu2011Physica,mostavi2008GB,corpet1988NAR,cole1996RAS,david1999Ties,pinterestEngineering,bateniGoogle,bateni2017NIPS}.
An HCT captures the idea that the network is a large cluster, which recursively splits
into smaller clusters of vertices, until only individual vertices remain in
each cluster. 
Thus, one key advantage is that we can learn the {\em private} HCT of the network once and then publicly release it. Any subsequent query on it for various levels of granularity does not require additional noise, due to the post-processing property.

To illustrate, in Figure~\ref{fig:dend} we draw the HCT for the sub-community from our motivating example (Figure~\ref{fig:motivating_example}). Each internal node of the HCT splits the network into two communities, one comprising of the leaves in the left subtree and the other from the right subtree. For instance, at the root the whole network is split in to two communities with Alice ($v_2$) in the left one and Bob ($v_6$) on the right. Observe that with the increasing depth of the HCT we get users in closer communities. To avoid confusion, we will henceforth refer to constituents of the HCT as nodes/links and that of the original network as vertices/edges.

Let $G:(V, E)$ be the original network. An HCT for $G$ has to group the vertices in $V$ by a measure of similarity/dissimilarity. Different applications can define different similarity notions. For
example, one natural definition is based on neighbor information---two vertices in $V$ are highly
similar if they have many common neighbors in $G$.  Degree similarity states that two nodes with a
larger difference in degrees are more dissimilar. Modularity measures the similarity of two clusters by
counting how many more edges between them exist than that predicted in a random
graph~\cite{newman2006Modularity}.
It is easy to see that members of an isolated clique in $G$ will be highly similar
to each other by all these definitions. 
In this work we design all our techniques using {\em
dissimilarity} scores; however, our techniques can be extended to the analogous
measures of similarity.

Generically, an HCT ($\tree$) has two components, $\Internal$ and $\theta$.
$\Internal$ is the set of internal nodes and $\theta$ is a dissimilarity measure defined as a
function that maps internal nodes to a real value. Each internal node $\inode_i\in\Internal$ is
associated with two subtrees $\lst_i$ and $\rst_i$. Since the leaves of each subtree represent the
vertices of the original graph, each internal node represents two sets of clusters, corresponding to
the left and right subtrees. Therefore $\theta(\inode_i)$ can also be thought of as the
dissimilarity score between two clusters represented by $\inode_i$. 
For instance, in Figure~\ref{fig:dend} the internal node $\inode_3$ has two clusters $\{(v_1, v_2), (v_3, v_4)\}$. The
$\theta(\inode_3)$ represents the dissimilarity between the two clusters at $\inode_3$.
Intuitively, the clusters at lower levels of the hierarchical structure should be more similar to
each other than at the higher levels i.e., closer to the tree root. This can be seen in the graph
$G$ (sub-community) as well, as $v_1, v_2, v_3, v_4$ are better clustered together than the other two vertices $v_5,v_6$.

Therefore, an HCT is meaningful only if the dissimilarity score at an
internal node $\inode$ is lower than the dissimilarity score at all of its ancestors of $\inode$. We
call this property as {\em ideal clustering property}. Next, we describe traditional algorithms to learn HCTs in the centralized setup with and without differential privacy, and the challenges of extending them to the federated setup with LDP.

\subsection{Learning an HCT in the Centralized Setup}
\label{sec:baseline}

In the centralized setup, 
the authority is trusted and knows the entire graph, i.e., it can query the raw edges in $G$ directly. 
The centralized setup offers a reasonable private baseline to compare our eventual solution to learn an HCT in the federated setup. 

Even with no DP guarantee, how do we compute $\tree$? Observe that there are combinatorially many (in $n$) cluster trees possible for $G$, since each cluster tree corresponds to a unique way of partitioning the vertices in the graph. The goal is to find a tree which preserves the ideal clustering property, and among those which do, find the one which provides the best clustering at each level.
Many traditional methods like average linkage~\cite{murtagh2012DMKD}, which are natural to use in the centralized setup, are ad-hoc and do not provide any quantitative way of measuring the quality of trees produced. A more systematic way is the algorithm
proposed by Clauset et al.~\cite{clauset2008Nature,clauset2007HRG} which we refer to as the \baseline algorithm. This algorithm is one of the most popular for computing hierarchical structures~\cite{hrgapp1, hrgapp2, hrgapp3}, and its DP version is known~\cite{qian2014KDD}.

The main idea in \baseline is to define a {\em cost function}, which quantifies how good is a cluster tree at clustering similar nodes at each level. For a given graph $G$ and a  probability assignment function $\pi$, the cost of a computed cluster tree $\tree$ is as follows:
\begin{align*}
\mle(\tree) &= \prod_{i=1}^{n-1}\left
(\pi_{i}\right)^{E_{i}}\left(1-\pi_{i}\right)^{L_{i} R_{i}-E_{i}}\text{, where }\pi_i = \pi(\inode_i)
\label{eq:globalcost}
\end{align*}

Here, $L_i$ and $R_i$ are the number of leaves in the left and right sub-trees at the $i$-th internal node $\inode_i$ in $\tree$. $E_i$ is the number of edges between the two clusters represented by $\inode_i$. The probability function $\pi$ assigns a probability score at each internal node $\inode_i$ which signifies the probability of an edge existing between a leaf (vertex) in the left subtree at $\inode_i$ and another vertex in the right. When the graph is available, $\pi_i$ is computed as $\pi_i = \frac{E_{i}}{L_{i}\cdot R_{i}}$. We will refer to edges that have one node in the left sub-tree and one in the right sub-tree as edges "crossing the clusters" rooted at $\inode_i$. The function $\pi_i$ gives a probabilistic interpretation to the cost function---it is the probability of an edge crossing the clusters rooted at $\inode_i$, for any graph (not necessarily $G$) that can be sampled using $\tree$ and $\pi$. Therefore, $\mle$ is the {\em likelihood} of sampling a graph using $\tree$ and $\pi$. 

A tree that optimizes this cost function given the underlying graph $G$, will sample $G$ that has the {\em maximum likelihood} among all graphs with $n$ nodes. Therefore, the \baseline algorithms employs the principle of maximum likelihood estimation to find a $\tree$ and $\pi$ conditioned on the given graph $G$ as evidence.
Maximizing $\mle$ by enumeratively evaluating it on the space of all possible cluster trees is intractable. Therefore, the \baseline algorithm optimizes for $\mle$ using a Markov chain Monte Carlo (MCMC) sampling procedure~\cite{mitzenmacherUpfal}. The MCMC sampling procedure is shown to converge in expectation to the desired $\pi$ for which $G$ maximizes the likelihood. 

\paragraph{\baseline-DP}The DP version of this algorithm is the state-of-the-art solution in the computation of differentially private
hierarchical structures~\cite{qian2014KDD}. This is our {\em centralized DP} baseline method and we call it
\baseline-DP. Specifically, the \baseline-DP simulates the exponential mechanism of differential
privacy~\cite{privacybook} by following a similar MCMC procedure to maximize $\mle$, and adds noise
to the edge counts $E_i$ after convergence for computing the probabilities $\pi$. Details of this
algorithm are elided here; we refer interested readers to prior work~\cite{qian2014KDD}. The key point is that both the
computation of \baseline and \baseline-DP assume access to the raw edge information of $G$. Now we ask, {\em can we extend \baseline-DP to the LDP setup?}

\subsection{Challenges: from \baseline-DP to LDP}
The cost function $\mle$ depends on the probability assignment function $\pi$ to measure the ``closeness'' between two clusters and $\pi$ uses fine-grained private information such as computing the exact number of interconnecting edges between the two clusters. To compute this fine-grained information in the federated setup, the users from one of the clusters can be asked to report the counts of their neighbors in the other cluster after adequately noising them for satisfying LDP. For instance, in Figure~\ref{fig:dend}, the number of edges ($E_3$) crossing the internal node $\inode_3$ can be calculated using the edge counts reported by users $v_1, v_2, v_3, v_4$ that cross $\inode_3$. However, searching for an  optimal $\tree$ that optimizes for $\mle$ requires computing the interconnecting edges for all possible sets of clusters in the worst case. Concretely, \baseline-DP implements an iterative search procedure that queries a user for edge information {\em thousands} of times before converging to the optimal $\tree$.  While this search works in the centralized setup where the graph is available, it is not feasible in the federated setup for two reasons. First, answering every differentially private query leads to a privacy loss and over many such queries the aggregated privacy loss will be high~\cite{privacybook}. \newtext{To exemplify, a thousand queries with an $\epsilon=0.1$ for each query will lead to an aggregated epsilon $\epsilon> 20$ with $99\%$ probability, as given by the advanced composition theorem~\cite{privacybook}. We consider $\epsilon \leq 2$ as reasonable~\cite{jagielski2020auditing}~\footnote{Usually, an $\epsilon=2$ allows an attacker to infer a random bit of the training sample (in this case an edge for each user) with $86\%$ probability.}, although prior works have considered up to $\epsilon=8$~\cite{abadi2016deep}.} Second, the search procedure for finding the optimal $\tree$ is iterative so the users have to be available for many iterations of the search to compute the closeness.

\section{Our Solution}
\label{sec:solution}

Our key ideas to tackle the aforementioned challenges are two-fold. First, instead of using a fine-grained probability assignment function $\pi$, we propose a coarse-grained method to compute the closeness between two clusters   such that the users need not be queried repeatedly for their neighbors. Second, we propose to replace the cost function $\mle$ with another cost function that can work with any coarse-grained method that captures the closeness between two clusters. These two observations enable us to design a novel hierarchical clustering algorithm which only queries the users once and the rest of the iterative search for an optimal tree happens on the server side. We detail our insights next.

\subsection{Key Insights}
\label{sec:approach}

Our first insight is to use a good coarse-grained approximation for closeness that can be obtained for a cheap privacy budget. We start from a construct called degree vectors, previously proposed in the LDP setup~\cite{qin2017CCS}. The degree vectors are a generalized version of degree counts, wherein the vertices are randomly partitioned into $\numbins$ bins (each bin has at least one vertex), and each user is asked to report how many neighbors  it has in each bin. For $\numbins=1$ this degree vector has one element and it yields just the degrees of the vertices in $G$. For $\numbins=n$, all nodes would have a unique bin, therefore the degree vectors will encode the original neighbor list for all vertices. For any $1<\numbins<n$ the idea is to preserve more edge information than just degrees and less than the exact edges. Then, the degree vector is noised by the user; a random Laplacian noise $Lap(0,\frac{1}{\epsilon})$ is added to each bin count before sending it to the untrusted authority. The $\numbins$ is usually a small value compared to the size of the network, typically $\leq \log(n)$ or a small constant~\cite{qin2017CCS}. This is intuitively good, since the noise added to the degree vector is proportional to $\numbins$.
Now we ask: What can we compute with degree vectors?
 
\begin{figure}%
  \centering
    \begin{subfigure}{0.18\textwidth}
      \includegraphics[width=\linewidth]{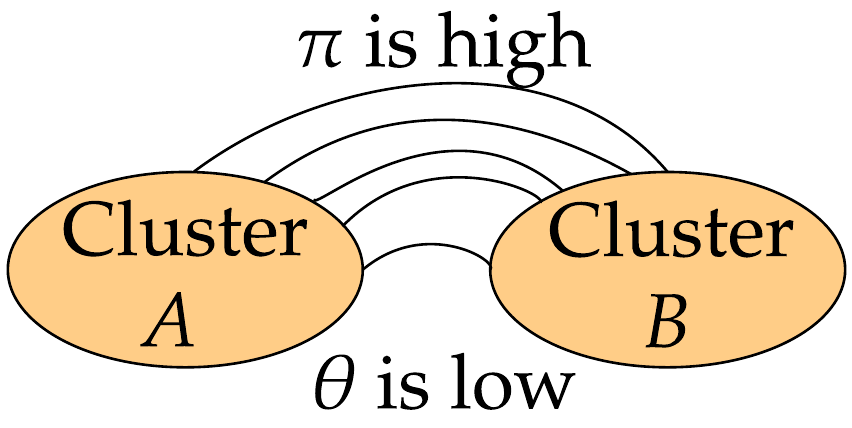}
      \caption{Close clusters} \label{fig:insight-a}
  \end{subfigure}%
  \hspace{3mm}
  \begin{subfigure}{0.18\textwidth}
    \includegraphics[width=\linewidth]{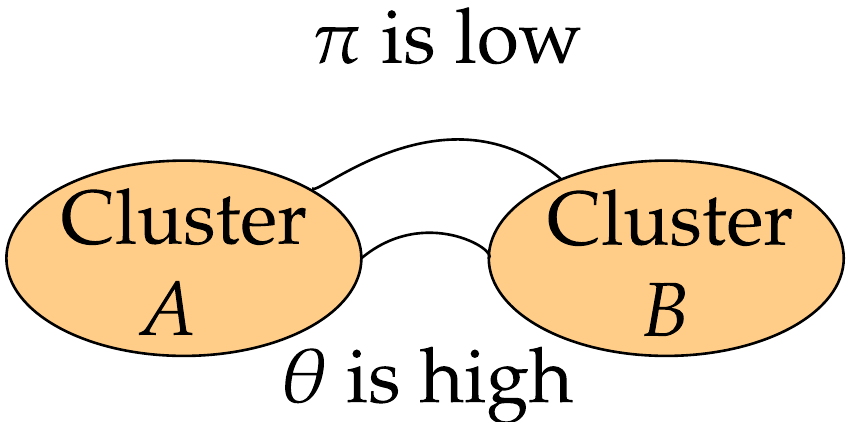}
    \caption{Distant clusters} \label{fig:insight-b}
  \end{subfigure}
  \caption{(a) Two close clusters, high $\pi$, have more interconnecting edges than two distant ones. (b) Average dissimilarity $\theta$ using \normL of degree vectors also captures closeness.}
  \label{fig:insight}
\end{figure}

It is not straightforward to compute $\mle$ using degree vectors. Nevertheless, observe that 
if we take two close clusters with respect to $\pi$, then on average each user in the left cluster has a lot of neighbors in the right cluster and vice versa. This notion is readily captured by degree vectors. Two vertices have very similar degree vectors if they have many common neighbors. Consequently, if we measure the dissimilarity as a \normL between their degree vectors then the dissimilarity for such vertices should be low. Therefore, the average dissimilarity across vertices of two {\em close} clusters will be lower than for two far clusters.
We illustrate our intuition in Figure~\ref{fig:insight}.
Figure~\ref{fig:insight-a} shows two clusters with high interconnecting edges, hence the $\pi$ associated with the two is high. Consequently, the average \normL distance, will be less. A similar argument can be made for Figure~\ref{fig:insight-b}. Therefore, we propose using the average \normL distance using degree vectors as a good coarse-grained replacement for $\pi$ to measure closeness between clusters. Observe that the degree vectors can be constructed just once for each user which can then be used to compute the average \normL distance for any two clusters. This key insight allows us to propose a different cost function that optimizes for the average \normL distance.

We point to a recently introduced cost function by Dasgupta~\cite{dasgupta2016STOC}.  Dasgupta's cost function takes a dissimilarity matrix and a tree as inputs and measures the quality of the tree for that dissimilarity matrix. Our first observation is that this cost function does not need the fine-grained edge counts between clusters (like $\mle$ does). The second and even more important observation is that a tree that optimizes Dasgupta's cost has a specific dissimilarity measure $\theta$ (defined in Section~\ref{sec:hct}). The measure represents the average dissimilarity between the vertices of two clusters computed using the dissimilarity matrix.
In fact, we formally show that such a tree satisfies the ideal clustering property as well (Section~\ref{sec:why-dasgupta}).
Thus, if we use the \normL of degree vectors to compute the dissimilarity matrix, then by optimizing for the Dasgupta's cost we get our desired coarse-grained average \normL distance as the $\theta$.
So, we directly arrive at a tree which minimizes the distance between the close clusters and maximizes the distance between the far ones.
	 
Using these insights, we design a novel randomized algorithm to sample a differentially private HCT that optimizes Dasgupta's cost function after querying a degree vector from each user. To summarize, we have shown a way to avoid multiple privacy-violating queries in learning an HCT by designing another learning strategy that allows us to replace the fine-grained queries with a coarse-grained one that preserves the ideal clustering property.

\subsection{Dasgupta's Cost Function}
\label{sec:why-dasgupta}
\begin{algorithm}[t]
\SetAlgoLined
\SetKwInOut{Input}{Input}
\SetKwInOut{Output}{Output}
\Input{Dissimilarities matrix $S$}
\Output{Hierarchical cluster tree $T^*$}
Randomly sample a tree $T_0$\;
$i=0$\;
 	\While{not converged}{
    Sample an internal node $\inode$ from $T_i$\;
    Construct two local swap trees $T_{i+1}^L, T_{i+1}^R$\;
    Pick one of them at random, call it $T_{i+1}$\;
    Transition to $T_{i+1}$ with $\min(1, e^{\dasgupta(T_{i+1})-\dasgupta(T_{i})})$\;
    $T^*$=$T_i$\;
    $i=i+1$\;
 	}
 	\Return $T^*$
 \caption{The outline of \localalgo.}
 \label{algo:localalgo}
\end{algorithm}

We now formulate the Dasgupta's cost function 
considering only full binary trees, and when a dissimilarity matrix is given as input. 

\begin{definition}[Dasgupta's Cost Function ($\dasgupta$)]
The cost of a tree $T:(\Internal, \costinode)$ with respect to a graph $G$ with
a non-negative dissimilarity matrix $S$ is given by
\begin{align*}
\mathcal{C}_{\mathcal{D}}(T) = \sum_{x\in V, y\in V} S(x,y)\cdot\abs{\leaves{T[x\lor y]}}\\
\text{ where }T[x\lor y]\text{ is least common ancestor of }x, y.
\end{align*}
\label{def:dasgupta}
\end{definition}

The expression is a weighted sum of dissimilarities of each pair of vertices, weighted by the number of leaves in the subtree of the least common ancestor of the pair of vertices. The idea of maximizing $\dasgupta$ is intuitive.
Any pair of vertices $(x, y)$ which are highly dissimilar to each other will have a
high dissimilarity score $\simscore{x}{y}$ and therefore, their least common ancestor
node (which we denote by $\lca$) should have more leaves in order to maximize
the product $\simscore{x}{y}\cdot|\leaves{\lca}|$.
It has been shown that, an optimal tree with respect to $\dasgupta$  has several nice properties such as: 
1) Two disconnected components in the graph will be separated completely into
two different clusters; 2) The cost of every tree for a clique is the same; and, 3) it
represents the clusters well in the planted partition model~\cite{dasgupta2016STOC}. Property 2) is the most relevant to us, as we will use it to bound the utility loss.

Originally, the cost function was designed for a non-negative similarity matrix with no restriction on the tree space. In the dissimilarity case, if we allow all possible trees then there is always a trivial tree that maximizes this function, i.e., the star graph with only one internal root node. Nevertheless, all the above properties hold in the dissimilarity case when the tree space is restricted to the full-binary trees~\cite{addad2018SIAM}.

More importantly, notice that the cost function does not require any edge information
of the graph but only a dissimilarity matrix, a key requirement for the design
of our algorithm in the decentralized setting.As stated previously, we show  that the tree $\tree$ that maximizes $\dasgupta$ also
satisfies the \idealprop. Further, the $\costinode$ of such a tree will just be
the average dissimilarity between the nodes of two subtrees at each internal node $\inode$. Formally, Theorem~\ref{thm:t1} captures both these observations.

\begin{restatable}{theorem}{thma}
\label{thm:t1}
Given a graph $G:(V, E)$ and a dissimilarity matrix $S$, the tree
$T^{OPT}:(\Internal,\costinode)$ that maximizes $\dasgupta$ preserves the \idealprop with
\begin{align*}
  \fcostinode{\inode} &= \frac{\sum_{x\in \lst, y \in \rst} \simscore{x}{y}}{|\lst|\cdot|\rst|}\\
                      &\text{ where $\lst$, $\rst$ are the left and right subtrees of $\inode$.}
\end{align*}
\end{restatable}

The full proof is provided in the Appendix~\ref{sec:proofs}. The theorem is proved by contradiction. Assuming that there exists a $\dasgupta$-optimal tree $T$ with two internal nodes $\inode_1, \inode_2$; $\inode_2 = \ancestor(\inode_1)$ such that $\theta(\inode_1)<\theta(\inode_2)$, we can always construct another tree $T'$ with a higher cost than $T$. In fact, if $T$ were the configuration $L$ in Figure~\ref{fig:localswap} then $T'$ would be one of the other configurations. 

Theorem~\ref{thm:t1} constitutes our first key analytical result. It explains why we choose Dasgupta's cost function and degree vectors together. Next, we describe our algorithm \localalgo to learn a tree that maximizes this cost function. \localalgo is an independent contribution of this work which can be used with any dissimilarity metric that captures closeness between clusters.

\begin{figure}
  \centering
  \resizebox{\linewidth}{!}{
    \includegraphics[width=\textwidth]{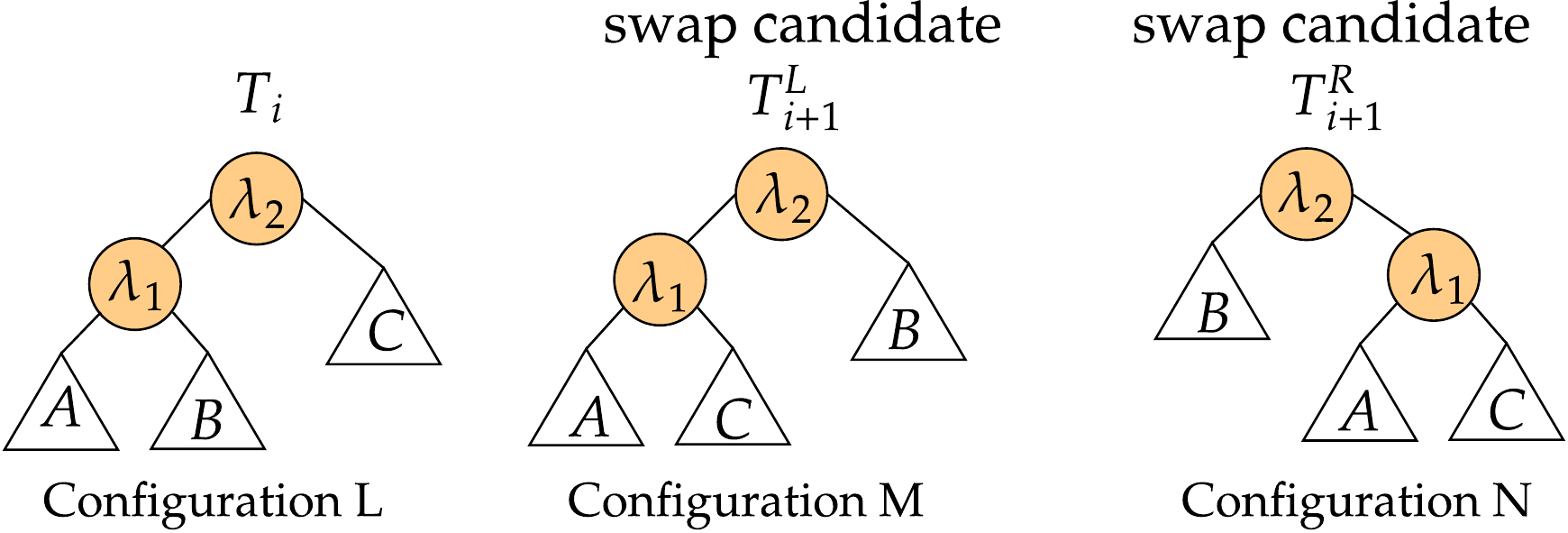}
  }
  \caption{The possible swaps of a tree.}
  \label{fig:localswap}
\end{figure}

\subsection{The \localalgo Algorithm}
\label{sec:hct-algo}

Finding a hierarchical cluster tree that maximizes $\dasgupta$ is known to be
NP-hard~\cite{dasgupta2016STOC}. Let $\alltrees$ be the set of all possible full binary trees that have $V$ as leaves. We want to find the optimal tree $\opttree \in \mathbb{T}$ that maximizes $\dasgupta$.
The number of possible trees including all the permutations from internal nodes
to leaves is at most $2^{cnlogn}$, where $c=O(1)$ is a small constant.

We propose a randomized algorithm \localalgo which will sample from the
distribution of all possible trees such that the sample probability is
exponential in the cost of the tree. Therefore, if the cost of the
tree is high then the sample probability is exponentially high. This
distribution is known as  Boltzmann (Gibbs) distribution and we choose it since it enables bounding the utility loss as we show later.
Hence, \localalgo samples a tree $T'\in \alltrees$ with probability
\begin{align*}
Pr(T'\in \alltrees) = \frac{e^{\dasgupta(T')}}{\sum_{T''\in \alltrees}{e^{\dasgupta(T'')}}}
\end{align*} 

\localalgo creates samples from this distribution by using \newtext{a Metro-polis-Hastings (MH) algorithm based} on a Markov chain with states as trees and state-transition probabilities as ratio of
costs between the states. The outline for \localalgo is given in Algorithm~\ref{algo:localalgo}. First, it starts with a randomly sampled tree
$T_0$. In each iteration, it samples a random internal node
$\inode$ and does a local swap into one of the configurations as shown in Figure~\ref{fig:localswap}. Observe that in a local swap, a subtree at the chosen internal node is detached and swapped with the subtree at the parent internal node leading to two possible local swaps. \newtext{This state transition (swap) between the trees $T_i, T_{i+1}$ is done with a
probability $\min(1, e^{\dasgupta(T_{i+1}) - \dasgupta(T_{i})})$ which is the ratio of the costs of the two states. This follows the standard MH procedure.}

Our chosen stationary distribution ensures that every state has a positive probability of being sampled. Further, every full binary tree can be obtained from another by a sequence of swaps therefore the entire chain is connected. Consequently, a standard analysis~\cite{mitzenmacherUpfal} leads to the conclusion that the Markov chain induced by \localalgo is ergodic and reversible. Therefore, \localalgo will {\em converge} to its stationary distribution, which in our case is the Boltzmann distribution, given enough time. We present the differentially private version of \localalgo next.

\subsection{The \dplocalalgo Algorithm}
\label{sec:dpgentree}

\newtext{\localalgo only requires a dissimilarity matrix to operate on. To construct a private hierarchical cluster tree using \localalgo we can first construct private dissimilarity matrix. Recall that any computation on a differentially private output is also differentially private using the post-processing property. Using the same property, we can construct a private dissimilarity matrix from private degree vectors. For that purpose, the central aggregator first sends a random partition of the users, with each bin having a set of users, to every user. Using these bins, every user constructs a private degree vector by counting their neighbors in each bin and adding Laplacian noise
$Lap(0,\frac{1}{\epsilon})$ to these counts. Finally, the users send their vectors to the aggregator so that the dissimilarity of every pair of users can be computed as measured by the \normL of their
respective degree vectors.  The outline of \dplocalalgo is summarized in the Algorithm~\ref{algo:localalgodp}. Notice that, the aggregator just requires the users to compute degree vectors privately so that \localalgo can be run on the server side to compute a private hierarchical cluster tree.} Further, \dplocalalgo requires a transfer of $O(n)$ bits of information between each user and the authority. 

\begin{algorithm}[t]
\SetAlgoLined
\SetKwInOut{Input}{Input}
\SetKwInOut{Output}{Output}
\Input{Graph $G$, dissimilarities matrix $S$}
\Output{Hierarchical cluster tree $T^*$}
    \textbf{Aggregator}: Randomly partition $G$ to $\numbins=\lfloor\log{n}\rfloor$ bins\;
	\textbf{Aggregator}: Show the $\numbins$ partitions to the user\;
	\textbf{User}: Send DP degree vectors with $Lap(0,\frac{1}{\epsilon})$ noise\;
	\textbf{Aggregator}: Compute dissimilarities $S$ using \normL\;
	\textbf{Aggregator}: Compute $T^*=\text{\localalgo}(S)$ and release $T^*$
 \caption{Outline for \dplocalalgo}
 \label{algo:localalgodp}
\end{algorithm}

 \section{Theoretical Bound On Utility Loss}
\label{sec:utilityloss}
\localalgo and \dplocalalgo are randomized algorithms. So, we need to show that they learn trees that are close to the ideal tree that optimizes $\dasgupta$ given the dissimilarity matrix computed using non differentially private degree vectors. The utility loss is given by the cost difference between the ideal tree and the trees sampled by our algorithms. Therefore, we bound the following:
\begin{enumerate}
	\item \localalgo: The utility loss for the tree output by the algorithm when dissimilarity matrix is not noised.
	\item \dplocalalgo: The utility loss for the tree output by the algorithm when the dissimilarity matrix is noised.
\end{enumerate}

For ease of analysis, we enforce the dissimilarity between two vertices of the graph to be at least $1$\footnote{We enforce this in our implementation for correctness.}, i.e., the \normL between two degree vectors is at least $1$ instead of $0$. Therefore, a clique will have a dissimilarity of $1$ for all pairs of nodes.
Following this, we state a fact that is proved in the original Dasgupta's work~\cite{dasgupta2016STOC}. 

\begin{restatable}{theorem}{thmc}
\label{thm:t0}
The cost of a clique with all dissimilarities $1$ is same for all possible trees and is equal to
\begin{align*}
\dasgupta(T_{clique}) = \frac{\numnodes^3 - \numnodes}{3}
\end{align*}
\end{restatable}

This is the least optimal cost tree that can be produced with any dissimilarity matrix under our assumption, i.e., $\simscore{v_i}{v_j}\geq 1$. For real-world graphs, the optimal tree cost can be many times higher than $\dasgupta(T_{clique})$ depending on the dissimilarity matrix. This is  useful in our proofs later. From here, we represent $\dasgupta(T_{clique})$ with $\rho$ for brevity. We now analyze the utility loss for \localalgo.

\subsection{Utility Loss: \localalgo}
There is always an ideal full binary tree that maximizes the Dasgupta cost function given the dissimilarity matrix, say $\opttree$ with $OPT=\dasgupta(\opttree)$. The \localalgo is a randomized algorithm which samples different tree $T$ at convergence in every run. Therefore, the expected utility loss is the difference between $\optcost$ and the expected value of the cost of obtained tree $\Expn{T}{\dasgupta(T)}$. We now show that the expected $\dasgupta$ cost of the sampled tree at convergence is close to $\optcost$ and  is concentrated around its expectation. 

The expected cost of a sampled tree is less than $\optcost$ only by at most a small factor $\frac{c\cdot log(n)}{n^2-1}$ of the $\rho$, where $c \ll n$. The complete proofs for the theorems are provided in the Appendix~\ref{sec:proofs}. Here, we explain the key ideas used in the proofs.

\begin{restatable}{theorem}{thmk}
\label{thm:t10}
\localalgo outputs a tree $T$ whose expected cost $\Expn{T}{\dasgupta(T)}$ is a
$(1-\frac{log(n)}{n^2-1})$ multiplicative factor of  $\optcost$.
\begin{align*}
\Expn{T}{\dasgupta(T)}&\geq (1-c\cdot \frac{log(n)}{n^2-1})\cdot \optcost
\end{align*}
\end{restatable}

To prove this, we first show that the probability of sampling a sub-optimal tree is exponentially decreasing.
\begin{restatable}{lem}{lemb}
\label{lem:l2}
The probability of sampling a tree $T$ with cost $\optcost - c'\cdot\numnodes\cdot\log(\numnodes)$ decreases exponential in $n$.
\begin{align*}
Pr(\dasgupta(T)\leq \optcost - c'\cdot\numnodes\cdot\log(\numnodes)) &\leq e^{-\numnodes\cdot\log(\numnodes)}
\end{align*}
\end{restatable}

Lemma~\ref{lem:l2} implies Theorem~\ref{thm:t10} and it also says that the utility loss is concentrated around its expectation. 
This is a consequence of choosing Boltzman distribution as the stationary distribution for \localalgo. The probability of sampling any sub-optimal tree exponentially reduces with its cost distance from the optimal cost.

\subsection{Utility Loss: \dplocalalgo}
\dplocalalgo samples trees based on the probabilities that are exponential in the cost of tree computed on the differentially private dissimilarities. Hence, the probability of sampling a tree $T$ will now depend on the noisy cost, say $g(T) = \dasguptabar(T)$ instead of the actual cost $f(T)=\dasgupta(T)$.
Let the variables $\simmatrix$ and $\dpsimmatrix$ be the dissimilarity and noisy dissimilarity matrices that store dissimilarities between vertices of the network. Recall that the dissimilarity between two vertices is computed as \normL of their degree vector counts. The cost computed with original dissimilarity matrix is $\actcost{T} = \dasguptafull{S}{T}$ whilst the cost computed by the noisy dissimilarity matrix is $\dpcost{T} = \dasguptafull{\overline{S}}{T}$. If $T_{dp}$ is the sampled tree at convergence then the expected utility loss is computed by taking the difference between $\optcost$ and $\Expn{T_{dp}}{f(T_{dp})}$.  In essence, we are saying that the tree sampled with differentially private dissimilarity matrix should not be away from the optimal tree with respect to the cost computed using the non-noised or actual dissimilarities. Formally,

\begin{restatable}{theorem}{thmh}
\label{thm:t8}
Let $T_{dp}$ be the output of \dplocalalgo and $\rho=\minoptcost$ be the minimal cost of any hierarchical cluster tree. The expected utility loss of \dplocalalgo is concentrated around its mean with a high probability and is given by
\begin{align*}
\optcost - \Expn{T_{dp}}{f(T_{dp})} \leq \frac{2\numbins}{\epsilon}\cdot\left(\frac{3}{2}+\frac{6}{\sqrt{\numbins}}\right)\cdot\rho
\end{align*}
\end{restatable}

In order to prove the bound, we start by first bounding the expected value of $g(T)$ over the randomness in degree vectors, in terms of $f(T)$. Let $\nrv_i$ denote the random variable that represents the Laplacian noise added by $i^{th}$ vertex to its degree vectors and $\nrv = (\nrv_1,\nrv_2,\ldots,\nrv_{\numnodes})$. 

\begin{restatable}{lem}{leme}
\label{thm:t5}
The expectation $|\Expn{R}{\dpcost{T}} - \actcost{T}|$ over randomness $R$ is bounded by
\begin{align*}
|\Expn{R}{\dpcost{T}} - \actcost{T}|\leq\frac{3\numbins}{2\epsilon}\cdot \rho\\
\end{align*}
\end{restatable}

Then we bound the variance of $g(T)$.

\begin{restatable}{lem}{lemf}
\label{thm:t6}
The variance of $\dpcost{T}$ is bounded by
\begin{align*}
   \Var_{R}[\dpcost{T}]\leq\frac{4\numbins}{\epsilon^2}\cdot\rho^2 
\end{align*}
\end{restatable}

Finally, we use the Chebyshev's inequality to bound the value of $g(T)$ in terms of $f(T)$.

\begin{restatable}{lem}{lemg}
\label{thm:t7}
The $\dpcost{T}$ is in the interval $\left[\actcost{T}-P, \actcost{T}+P\right]$ with a high probability where $P$ is,
\begin{align*}
    P = \frac{\numbins}{\epsilon}\cdot\left(\frac{3}{2}+\frac{6}{\sqrt{\numbins}}\right)\cdot\rho
\end{align*}
\end{restatable}

We then use this lemma to prove Theorem~\ref{thm:t8}. Note that the utility loss does not depend on the dissimilarity matrix but only depends on the parameters $\numbins$ and $\epsilon$. It can be seen that if $\epsilon$ is very high then the utility loss converges to that of the \localalgo and even if $\epsilon$ tends to zero, our artificial bounding of the dissimilarities to be greater than $1$ will ensure that the final tree cost never goes below $\rho$. Recall that $\optcost$ is several times more than $\rho$ for the real-world graphs as it depends on the dissimilarity matrix, i.e., the structure of the graph.  We empirically observe that the utility loss is smaller than our bound for all of our evaluated graphs (see Section~\ref{sec:eval}).
 \section{Evaluation: Quality of private HCT}
\label{sec:eval}

In this section, our goal is to evaluate the utility of \dplocalalgo as measured by the quality of the HCTs it produces.

First, we want to evaluate the quality of the private HCTs that \dplocalalgo generates vs. the non-DP HCTs that \localalgo generates.
Our algorithms use a different cost function (Dasgupta's $\dasgupta$) than the centralized baseline that allows us to bound the utility loss theoretically (Section~\ref{sec:utilityloss}). Here, we measure the utility loss empirically with respect to the HCT generated by \localalgo as measured by $\dasgupta$.
Secondly, we want to evaluate the utility-privacy trade-off of the LDP regime vs. the centralized DP setup. Thus,  we compare the quality of the private HCTs that \tool generates with LDP guarantees and the trees generated by the centralized DP algorithm \globalalgo for the same privacy budget.
In order to compare these trees, we use the cost function $\mle$ which measures closeness between two clusters using the fine-grained edge counts as opposed to our coarse-grained degree vectors (see Section~\ref{sec:baseline}).
In summary, we ask the following:

\textbf{(EQ1)} What is the empirical utility loss of the private HCTs produced by \dplocalalgo vs. the non-DP HCTs produced by \localalgo?

\textbf{(EQ2)} What is the quality of the private HCTs produced by \dplocalalgo vs. the centralized DP algorithm \globalalgo?

\begin{figure*}
  \begin{tikzpicture}
   \begin{groupplot}[group style = {group size = 3 by 1, horizontal sep = 50pt},
     width=1/3.5*\textwidth,
     height = 4cm]
  \nextgroupplot[title={\lastfm},
    legend style={column sep = 2pt, legend columns = -1, legend to name = grouplegend,},
    xlabel={steps/$n$},
    ylabel={$\log{\mle}$},
    every axis plot/.append style={thick}
    ]
    \addplot[mark=*,mark phase=500,mark repeat=600,
    mark options={fill=brown,draw=black,scale=1.4},color=chocolate(traditional)] table [x=steps/n, y=loss,col sep=comma]
      {figures/data-rec/mcmc_losses/lastfm-small_05_mcmc_losses.csv};
    \addlegendentry{LDP,$\epsilon=0.5$}
    \addplot[mark=diamond*, mark phase=500, mark repeat=600,mark
    options={fill=blue,draw=black,scale=1.4},color=blue] table [x=steps/n, y=loss,col sep=comma]
      {figures/data-rec/mcmc_losses/lastfm-small_10_mcmc_losses.csv};
    \addlegendentry{LDP,$\epsilon=1.0$}
    \addplot[mark=triangle*,mark phase=500,mark repeat=600,mark
    options={fill=red,draw=black,scale=1.4},color=red] table [x=steps/n, y=loss,col sep=comma]
      {figures/data-rec/mcmc_losses/lastfm-small_20_mcmc_losses.csv};
    \addlegendentry{LDP,$\epsilon=2.0$}

    \addplot[mark=oplus*,mark phase=500,mark repeat=600,
    mark options={fill=magenta,draw=black,scale=1.4},color=magenta] table [x=steps/n, y=loss,col sep=comma]
      {figures/data-rec/mcmc_losses/global_lastfm-small_05_mcmc_losses.csv};
    \addlegendentry{DP,$\epsilon=0.5$}
    \addplot[mark=square*,mark phase=500,mark repeat=1000,mark
    options={fill=amethyst,draw=black,scale=1.4},color=amethyst] table [x=steps/n, y=loss,col sep=comma]
      {figures/data-rec/mcmc_losses/global_lastfm-small_10_mcmc_losses.csv};
    \addlegendentry{DP,$\epsilon=1.0$}
    \addplot[mark=pentagon*,mark phase=500, mark repeat=1000,mark
    options={fill=azure(colorwheel),draw=black,scale=1.4},color=azure(colorwheel)] table [x=steps/n, y=loss,col sep=comma]
      {figures/data-rec/mcmc_losses/global_lastfm-small_20_mcmc_losses.csv};
    \addlegendentry{DP,$\epsilon=2.0$}

\nextgroupplot[title=\delicious,
  xlabel={steps/$n$},
  ylabel={$\log{\mle)}$},
  every axis plot/.append style={thick}
  ]
  \addplot[mark=*, mark phase=500, mark repeat=600,
    mark options={fill=brown,draw=black,scale=1.4},color=chocolate(traditional)] table [x=steps/n, y=loss,col sep=comma]
    {figures/data-rec/mcmc_losses/delicious_05_mcmc_losses.csv};
  \addplot[mark=diamond*, mark phase=500, mark repeat=600,mark
  options={fill=blue,draw=black,scale=1.4},color=blue] table [x=steps/n, y=loss,col sep=comma]
    {figures/data-rec/mcmc_losses/delicious_10_mcmc_losses.csv};
  \addplot[mark=triangle*, mark phase=500, mark repeat=600,mark
  options={fill=red,draw=black,scale=1.4},color=red] table [x=steps/n, y=loss,col sep=comma]
    {figures/data-rec/mcmc_losses/delicious_20_mcmc_losses.csv};

  \addplot[mark=oplus*,mark phase=500, mark repeat=600,mark
  options={fill=magenta,draw=black,scale=1.4},color=magenta] table [x=steps/n, y=loss,col sep=comma]
    {figures/data-rec/mcmc_losses/global_delicious_05_mcmc_losses.csv};
  \addplot[mark=square*,mark phase=500, mark repeat=1000,
  mark options={fill=amethyst,draw=black,scale=1.4},color=amethyst] table [x=steps/n, y=loss,col sep=comma]
    {figures/data-rec/mcmc_losses/global_delicious_10_mcmc_losses.csv};
  \addplot[mark=pentagon*,mark phase=500, mark repeat=1000,mark
  options={fill=azure(colorwheel),draw=black,scale=1.4},color=azure(colorwheel)] table [x=steps/n, y=loss,col sep=comma]
    {figures/data-rec/mcmc_losses/global_delicious_20_mcmc_losses.csv};
\nextgroupplot[title=\douban,
  xlabel={steps/$n$},
  ylabel={$\log{\mle)}$},
  every axis plot/.append style={thick}
  ]
  \addplot[mark=*, mark phase=500, mark repeat=600,
    mark options={fill=brown,draw=black,scale=1.4},color=chocolate(traditional)] table [x=steps/n, y=loss,col sep=comma]
    {figures/data-rec/mcmc_losses/douban_05_mcmc_losses.csv};
  \addplot[mark=diamond*, mark phase=500, mark repeat=600,mark
  options={fill=blue,draw=black,scale=1.4},color=blue] table [x=steps/n, y=loss,col sep=comma]
    {figures/data-rec/mcmc_losses/douban_10_mcmc_losses.csv};
  \addplot[mark=triangle*,mark phase=500,mark repeat=600,
  mark options={fill=red,draw=black},color=red] table [x=steps/n, y=loss,col sep=comma]
    {figures/data-rec/mcmc_losses/douban_20_mcmc_losses.csv};

  \addplot[mark=oplus*,mark phase=500, mark repeat=600,mark
  options={fill=magenta,draw=black,scale=1.4},color=magenta] table [x=steps/n, y=loss,col sep=comma]
    {figures/data-rec/mcmc_losses/global_douban_05_mcmc_losses.csv};
  \addplot[mark=square*,mark phase=500,mark repeat=600,
  mark options={fill=amethyst,draw=black,scale=1.4},color=amethyst] table [x=steps/n, y=loss,col sep=comma]
    {figures/data-rec/mcmc_losses/global_douban_10_mcmc_losses.csv};
  \addplot[mark=pentagon*,mark phase=500,mark repeat=1000,mark
  options={fill=azure(colorwheel),draw=black,scale=1.4},color=azure(colorwheel)] table [x=steps/n, y=loss,col sep=comma]
    {figures/data-rec/mcmc_losses/global_douban_20_mcmc_losses.csv};
\end{groupplot}
\node at ($(group c2r1) + (0,-2.5cm)$) {\ref{grouplegend}};
\end{tikzpicture}

\caption{The \dplocalalgo (marked as LDP) and the \globalalgo (marked as DP) evaluated with
  $\log{\mle}$ at each step of the MCMC process for $3$ values of $\epsilon\in\{0.5, 1.0, 2.0\}$.
  Higher values of $\log{\mle}$ signify smaller loss, so better trees.  \dplocalalgo produces trees
  with a better $\mle$ cost than the baseline \globalalgo for all values. We also observe that \dplocalalgo
converges within $400$ MCMC steps.}
  \label{fig:mcmc_loss}
\end{figure*}
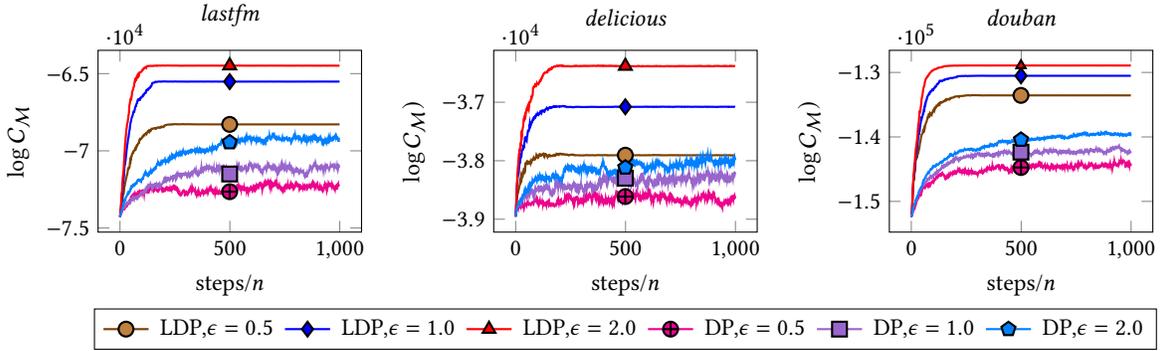

\begin{table}[tp]
\caption{Dataset graph statistics}
\label{tab:datasets}
\centering
\resizebox{0.9\linewidth}{!}{%
\begin{tabular}{|c|c|c|c|c|}
\hline
\textbf{Network} & \textbf{Domain} & \textbf{\#Nodes} & \textbf{\#Edges} & \textbf{Density} \\ \hline
\lastfm & social graph & 1843 & 12668 & 0.0075 \\ \hline
\delicious & social graph & 1503 & 6350 & 0.0056 \\ \hline
\douban & social graph & 2848 & 25185 & 0.0062 \\ \hline
\end{tabular}%
}
\end{table}

\begin{table}[t]
\caption{The quality of HCT evaluated using Dasgupta cost function. The utility loss is at most
\newtext{$10.87\%$} for a small $\epsilon=0.5$. }
\label{tab:dasgupta-loss-emp}
\begin{center}
\resizebox{0.35\textwidth}{!}{%
\begin{tabular}{cccc}
\textbf{Network} & \textbf{$\epsilon$} & 
\multicolumn{1}{c}{\textbf{\begin{tabular}[c]{@{}c@{}}Relative \\ Utility \end{tabular}}} & 
\multicolumn{1}{c}{\textbf{\begin{tabular}[c]{@{}c@{}}Empirical \\ Utility Loss (LDP) \end{tabular}}} \\
\hline
\multirow{3}{*}{\lastfm} & 0.5 & \multirow{3}{*}{\newtext{22.82}} & \newtext{9.57} \\
 & 1.0 &  & \newtext{4.05} \\
 & 2.0 &  & \newtext{1.45} \\ \hline
\multirow{3}{*}{\delicious} & 0.5 & \multirow{3}{*}{\newtext{12.20}} & \newtext{10.87} \\
 & 1.0 &  & \newtext{6.61} \\
 & 2.0 &  & \newtext{3.29} \\ \hline
\multirow{3}{*}{\douban} & 0.5 & \multirow{3}{*}{\newtext{28.83}} & \newtext{7.09} \\
 & 1.0 &  & \newtext{3.27} \\
 & 2.0 &  & \newtext{1.14} \\ \hline
\end{tabular}%
}
\end{center}
\end{table}

\label{sec:quality-hct-eval}

\paragraph{Experimental Setup}
We use $3$ real-world networks which are commonly used for evaluating recommender systems~\cite{wang2018collaborative,shepitsen2008personalized,ma2011recommender,qin2017CCS}~\footnote{https://grouplens.org/datasets/hetrec-2011/}.
We detail the networks in Table~\ref{tab:datasets}.
\newtext{We use privacy budget values of $\epsilon\in\{0.5, 1.0, 2.0\}$ as they have been used in the prior work for both centralized as well as the local differential privacy setup~\cite{hsu2014differential,qian2014KDD,abadi2016deep}.}
We have two tunable parameters: the number of bins $\numbins$ and the convergence criteria for MCMC in \localalgo. The $\numbins$ is chosen such that it minimizes the noise in the final degree vector as well as it minimizes the number of collisions in the degree vector calculation. A higher $\numbins$ minimizes the collisions but incurs more noise on aggregate and vice versa. Previous work has shown that the value of $\numbins$ is usually small and it depends on the graph structure. 
\newtext{We therefore heuristically choose $\numbins=\lfloor \log{n}\rfloor$ as it is not too low for not capturing any edge information (see Section~\ref{sec:approach}) and not too high so that we might end up adding a lot of noise (see Section~\ref{sec:prob_setup}). If we substitute this $K$ value in our utility bound for \dplocalalgo in Theorem~\ref{thm:t8}, we get a loss scaling with $\frac{\log{n}}{n\cdot\epsilon}$. Further, our values of $K$ agree with the ones used in the prior work~\cite{qin2017CCS}}. 
We use a convergence criteria that is similar to the ones used by prior works in the centralized setting for their MCMC algorithm~\cite{qian2014KDD,clauset2008Nature}.

\paragraph{Quality of LDP vs. non-DP}
We show that the observed cost of the differentially private tree generated by \dplocalalgo, for all values of $\epsilon$ is very close to the cost of the tree generated by the non-DP version of the algorithm \localalgo. 

We define the {\em empirical utility loss} as $\frac{|\dasgupta(T_{dp})-\dasgupta(T)|}{\dasgupta(T)}$, where $\dasgupta(T_{dp})$ is the cost of the private HCT and $\dasgupta(T)$ is the cost the non-private tree.
Recall that the cost of least optimal tree $\rho$ corresponds to that the tree obtained from a clique. Thus, while we cannot compute the optimal tree, we can compute its {\em relative utility} with respect to the cost $\rho$ as $\frac{\dasgupta(T)}{\rho}$. The higher this value, the better the non-private tree is.
In Table~\ref{tab:dasgupta-loss-emp}, we observe that the empirical utility loss is less than
\newtext{$10.87\%$} for all networks and values of $\epsilon$ we evaluated. The empirical utility loss is
approximately $0.08-1.67\%$ on average for $\epsilon=2.0$. This confirms that empirically
\dplocalalgo generates trees with a small loss of utility.
The empirical loss is within the provided conservative theoretical bounds, but, tightening the
theoretical bounds remains promising future work.

\custombox{1}{The empirical utility loss between the DP and non-DP version is less
than \newtext{$10.87\%$} for all the graphs and $\epsilon$ values. }

\paragraph{Quality of LDP HCTs vs centralized DP HCTs}
Our main result is that our LDP algorithm produces trees that have a better utility cost than the ones produced by the centralized DP algorithm as measured by the baseline $\mle$ cost.
Both our non-DP and LDP algorithms generate trees while optimizing for a different Dasgupta's cost function. 
In Figure~\ref{fig:mcmc_loss}, we show that the $\log{\mle}$ of the trees generated by \dplocalalgo
have better utility than trees generated by \globalalgo for all values of $\epsilon \in
\{0.5,1.0,2.0\}$ across all $3$ graphs.

\custombox{2}{The $\mle$ cost of the trees generated by our \dplocalalgo is better than the baseline
\globalalgo for all evaluated graphs and $\epsilon$ values.}

\section{Application: Recommender Systems}
\label{sec:case-study}

We revisit the motivating example of social recommendation from
Section~\ref{sec:problem} and ask ``Can \tool be utilized for providing
recommendations in the federated setup?''
The goal is to accurately predict the top-$k$ recommendations for the target
users. 
The target users share their preferences/tastes with the recommender system but
do not want to share their private contacts. Therefore, this naturally falls in
the edge-LDP setup.

We consider two scenarios: cold-start and existing users. Recall that the cold-start scenario occurs when a new user with no history of interactions with the service provider has to be served recommendations. In contrast, existing users have a history of interactions that convey their interests.
Specifically, we evaluate the utility of the \tool algorithm for social
recommendations by asking the following questions:

\textbf{(EQ1)} For cold-start, are the top-$k$ recommendations with
    \tool-based algorithm better than the non-social baseline?

\textbf{(EQ2)} For existing users, can we improve the top-$k$ recommendations over the
    non-social baselines with \tool-based technique?

\paragraph{Datasets} We use the same datasets mentioned in Section~\ref{sec:eval} as they are commonly used for social
recommendations.
For \lastfm and \delicious which are implicit feedback ratings we normalize the
weight of each user-artist and user-bookmark link by dividing the maximum number
of plays/bookmarks of that user across all artists/bookmarks. For \douban with
explicit feedback ratings we normalize the ratings by dividing with the maximum
rating ($5.0$).

\paragraph{Experimental Setup} We split the datasets in training and test sets.
We perform k-fold cross-validation for $k=5$ splits of the datasets.
For the HCT computation, we use a privacy budget of $\epsilon=1$ and we average
results over HCTs computed for $3$ random seeds.

\paragraph{Evaluation Metrics}
To evaluate our recommendations we use mean average precision (MAP) and normalized
discounted cumulative gain (NDCG). These two metrics are widely used to evaluate
the quality of the top-$k$ recommendations. Higher values of these metrics imply better recommendations. Due to space constraints, we explain
how these are computed in Appendix~\ref{sec:eval-metrics}.

Next, we describe our baseline and \tool-based algorithm,
and then present our results on our social recommendation datasets.

\subsection{Collaborative Filtering}
We implement a widely used technique for recommender systems called
collaborative filtering (CF)~\cite{netflixusesit,resnick1994grouplens}.
In particular, we consider memory-based CF techniques that operate on
user-item rating matrices to make top-$k$ recommendations.
We choose this technique because it generally yields good performance,
forms the basis of real-world recommender systems such as
Netflix~\cite{netflixusesit}, Facebook~\cite{fbusesit}.

\paragraph{CF for cold start}
In this scenario, the target users have no history of rated items. We choose two popular strategies, ``recommend highest rated items'' (\itemAvg) and ``recommend highest rated items among your friends'' (\friendsCF). Notice that these two methods correspond to two different granularities of the ``close'' community. \itemAvg considers the entire network where as the \friendsCF takes only the neighbors.

For \friendsCF, if $S:U\times U \rightarrow\{0,1\}$ is the social relation function and $S(u', u)=1$ if the
users $(u, u')$ are friends and $0$ otherwise, then the rating $r_{u,i}$ is computed as follows:
\begin{align*}
  r_{u,i} = \overline{r}_i + \frac{\sum_{u'\in \mathcal{U}} S(u',u)
  (r_{u',i}-\overline{r}_{u'})}{\sum_{u' \in \mathcal{U}}S(u',u)}
\end{align*}
where $\overline{r}_i=\frac{\sum_{u' \in \mathcal{U}} r_{u',i}}{|\mathcal{U}|}$
is the average rating of item $i$,
$\overline{r}_{u'} = \frac{\sum_{i \in \mathcal{I}} r_{u',i}}{|\mathcal{I}|}$ represents the
average item rating for user $u'$,
$\mathcal{I}$
is the set of items with a rating
from user $u'$ and
$\mathcal{U}\subseteq U$
is the set of users that have a
rating for item $i$. For \itemAvg, we just replace $S(u', u)$ with $1$ for all $u'\in \mathcal{U}$. Note that \itemAvg is a private baseline as no one shares their social relations where as \friendsCF is a non-private baseline.

\paragraph{CF for existing users}
When the user is present in the platform, we can leverage both their ratings {\em and} social
connections to make better recommendations.  In this case, we use one of the state-of-the-art social
recommendation algorithm called \serec~\cite{wang2018collaborative}.
\serec first models the user exposures to the items and then uses these exposures to guide the
recommendation. The exposures are computed based on the social information rather than the rating history.
Due to space constraints, we refer the reader to the paper for more details.
We use a popular Python implementation\footnote{https://github.com/Coder-Yu/RecQ} which follows the original implementation.
This represents the non-private baseline for social recommendations for existing users.
\subsection{Collaborative Filtering with \dplocalalgo}

The close community to which a target user belongs can be queried from an HCT which is typically the clusters at the lower levels of the HCT. We can control the granularity of ``closeness'' by the popular query ``Who are the $m$~\footnote{popularly known as $k$ nearest neighbors. We use $m$ because $k$ is used elsewhere. } nearest neighbors for user $u_i$?''. This can be answered with a set $N(u_i):(n_1, n_2, ..., n_m)$ where $n_1$ belongs to the closest community of $u_i$, $n_2,...,n_{j<m}$ in the second closest community and so on. This query forms the main insight for recommendations based on \dplocalalgo.

We propose a simple strategy to replace the privacy invasive algorithms that depend on direct social relations into a differentially private one. Specifically, we query the \dplocalalgo-based HCT to give the $m$ nearest neighbors to the target user $u_i$ and use them for recommendations.
We use this strategy to replace our baselines with and without cold-start.

\begin{enumerate}
    \item For \friendsCF, we select $m$ closest users $N(u_i)$ for each target user $u_i$ based on HCT and we employ the formula used by \friendsCF with $S(u',u)=1.0$ when $u'\in N(u_i)$.
    \item For SERec, we do the same as above i.e., replace the immediate social contact with the closest users in HCT and feed it to the algorithm.
\end{enumerate}

For both the above strategies we choose $m$ as the degree of the target user $u_i$ in order to consider similar number of neighbors to the baseline strategies to be fair. Recall that the degree of a user can be estimated from its degree vector we computed for \dplocalalgo. However, in the real world $m$ is application dependent and might require domain expertise or online testing to estimate an ideal value. Note that our insights can be used in any algorithm that depends on social contacts for recommendations. Furthermore, recommendation algorithms can be designed using the private HCT with the availability of additional information such as node attributes. We leave that for future work.

We are aware of only one prior work that aims to design social recommendation algorithms in the LDP setup. LDPGen, which appeared at CCS 2017, uses differentially private synthetic graphs for recommendations~\cite{qin2017CCS}. It uses the direct graph edge information in the synthetically generated graphs for recommendations, unlike our work. The data structure they compute (synthetic graph) is fundamentally different from  ours (hierarchical cluster tree) and therefore, their recommendation algorithm as well. 
However, a direct comparison to LDPGen is not possible. LDPGen does not accompany any publicly available implementation. After our several attempts to implement it faithfully, we were unable to reproduce findings reported therein. Furthermore, the presented theoretical analysis in the LDPGen paper has flaws, which were confirmed in private communication with one of the authors of LDPGen. We detail both our experimental effort and observed theoretical inconsistencies of that work in Appendix~\ref{sec:ldpgen}.

\subsection{Performance of \dplocalalgo-based CF}
\label{sec:case-study-eval}
\begin{table}[]
\caption{Results for the two cold start setup. \privaCTCF is a \newtext{$1-390\times$} better than the \itemAvg which shows the utility of using community information from the private HCT.}
\begin{tabular}{llccc}
                            &  Top-$100$    & \textbf{\itemAvg} & \textbf{\friendsCF} &
                            \textbf{\privaCTCF} \\ \hline
\multirow{2}{*}{\lastfm}    & NDCG & \newtext{5.33E-04}         & \newtext{1.58E-01}           & \newtext{6.21E-02}           \\ \cline{2-5} 
                            & MAP  & \newtext{2.56E-05}         & \newtext{4.03E-02}           & \newtext{8.63E-03}          \\ \hline
\multirow{2}{*}{\delicious} & NDCG & \newtext{8.80E-04}         & \newtext{6.70E-03}           & \newtext{9.00E-04}           \\ \cline{2-5} 
                            & MAP  & \newtext{7.18E-05}         & \newtext{8.00E-04}           & \newtext{6.59E-05}           \\ \hline
\multirow{2}{*}{\douban}    & NDCG & 2.00E-04         & 7.49E-02           & \newtext{7.800E-02}           \\ \cline{2-5} 
                            & MAP  & 1.46E-05         & 1.87E-02           & \newtext{2.000E-02}          \\ \hline
\end{tabular}
\label{tab:cold_start}
\end{table}

Our main finding is that in both evaluated scenarios, cold start and existing users, our \tool-based algorithms can be used for social recommendations with almost no utility loss for a privacy budget $\epsilon=1$. 
In both scenarios for all evaluated networks, \tool-based algorithms perform  \newtext{$0.9-390\times$} better \newtext{(as per NDCG)} than the \newtext{{\em non-social} recommendation baselines where recommendations are computed without social information}.
Compared to the non-private baselines \newtext{that do use the social links}, \tool is either on par with state-of-the art algorithm (\serec) in the existing users case or no worse than \newtext{$7.44\times$} for cold start \newtext{(for NDCG, or $12.1\times$ for MAP).}

For cold start CF, \privaCTCF outperforms the private baseline \itemAvg by
\newtext{$1.1-390\times$} across all datasets for NDCG and \newtext{$0.9-1.35\cdot10^{3}\times$} for MAP (Table~\ref{tab:cold_start}).
The reason is that the close communities extracted by \tool through the HCT are
a good proxy of similarities between users, whereas \itemAvg is a simple average
over the items.
For \lastfm, the non-private \friendsCF is \newtext{$2.54\times$} better in NDCG score than \privaCTCF and,
respectively, \newtext{$7.44\times$} better for \delicious. \newtext{Similarly, for \lastfm, the non-private \friendsCF is \newtext{$4.67\times$} better in MAP score than \privaCTCF and,
respectively, \newtext{$12.1\times$} better for \delicious.
For \douban, \privaCTCF is \newtext{$1.1\times$} better than the
\friendsCF in NDCG score, and \newtext{$1.1\times$} in MAP}.

\custombox{3}{For cold-start users,  \privaCTCF's recommendations are better than the \newtext{non-social} baseline for all networks across all metrics.}

\begin{table}[]
\caption{\privact-\serec is \newtext{$2.16-17.95\times$} better than the non-social baseline (Basic \serec) and comparable to the non-private baseline \serec.
}
\begin{tabular}{clccc}
                            &  Top-$100$    & \textbf{Basic \serec} & \textbf{\serec} &
                            \textbf{\privact-\serec} \\ \hline
\multirow{2}{*}{\lastfm}     & NDCG & 0.1223                   & 0.3249                 & \newtext{0.3270}                  \\ \cline{2-5} 
                            & MAP  & 0.0240                   & 0.1405                 & \newtext{0.1419}                \\ \hline
\multirow{2}{*}{\delicious} & NDCG & 0.0058                   & 0.0135                 & \newtext{0.0150}                  \\ \cline{2-5} 
                            & MAP  & 0.0021                   & 0.0048                 & \newtext{0.0045}                  \\ \hline
\multirow{2}{*}{\douban}    & NDCG & 0.0349                   & 0.2466                 & \newtext{0.2469}                  \\ \cline{2-5}
                            & MAP  & 0.0052                   & 0.0929                 & \newtext{0.0930}                  \\ \hline
\end{tabular}
\label{tab:serec}
\end{table}

For existing users, \privact-\serec is performing on par with the non-private
state-of-the-art algorithm \serec for top-$100$ recommendations for all $3$
datasets (see Table~\ref{tab:serec}). These results show that the similarity
between users as captured by the communities of our private hierarchical
clustering are a good indicator for trust within the direct connections in the
social network. \privact-\serec has better \newtext{$2.67-7.08\times$} predictions than the \serec \newtext{without social information in NDCG score, and, respectively \newtext{$2.16-17.95\times$} better in MAP}. \serec (with social information) has similar performance to \privact-\serec. 

\custombox{4}{For existing users, \privact-\serec has \newtext{$2.16-17.95\times$} better
predictions than the \newtext{non-social} \serec for all evaluated graphs across all metrics.
\privact-\serec has similar performance to \serec.}
\subsection{Why Does \privaCTCF Work Well?}
\label{sec:eval-observations}

Recall that a user is influenced by other users based on the distance between them on the network~\cite{golbeck2005computing,hang2008operators}.
Thus, for our first experiment, we measure the correlation between the similarity of users' rating profile and the distance between them on the network.
We find a negative correlation between these two in all of our datasets, i.e., similar users tend to be closer.
We next evaluate if a similar correlation can be observed in the top-$100$ recommendations produced by \privaCTCF. Specifically, we ask whether higher NDCG scores are correlated to shorter distances between the target user and users in their close community.
We observe that for $2$ of our datasets, there is, indeed, a negative correlation.
This means that for the top items, the HCT-generated close community mostly consists of users that are at a shorter distance on the network. In contrast, the community corresponding to less relevant recommendations, has larger distances among its users.
We next detail how we perform these experiments.

\paragraph{User Interests vs. Distance Correlation.}
There are two ways to measure the similarity between users. First, we consider the
partial profile of users, i.e., we consider only the items that have been rated by the users.
Second, we complete the rating profiles by giving a
$"0.0"$
rating to the non-rated ones.
This captures the fact that similar users not only have similar ratings for specific items but also rate the same items.
Finally, we compute the correlation between the shortest path distances and the similarities computed between the users.
All correlations are measured using the Pearson correlation.
The results are summarized in the Appendix~\ref{sec:appobs}, Table~\ref{tab:corr}.
For partial profiles, we observe that
the correlation between shortest path distances and user similarity is negative for $77\%, 96\%$ and
$95\%$ of the users in the three datasets \lastfm, \delicious and \douban.
For full rating profiles, the correlation is negative for all users.
This is an indication that as the distance increases between users, their similarity decreases on all our datasets, confirming our intuition that users are influenced based on distance on the network.

\paragraph{NDCG Scores vs. Close Communities}
Our goal is to measure whether relevant items that are ranked higher for a target user are recommended by other users that are at shorter distances from the target user on average. Recall that in the cold start scenario,  we choose $m$ users as the close community, where $m$ is the degree of the target user.
To measure how the shortest paths vary across the top-$100$ predictions using \privaCTCF, we first compute the shortest path between the target user and other users in their $m$-sized community who have rated the same items.
Thus, for each top-$100$ item of a particular target user, we obtain a corresponding average shortest path between users who have rated it.
We construct the {\em shortest path profile} for each user as an array that contains the average shortest paths corresponding top-$100$ rated items. 
Finally, we compute the Pearson's correlation between the NDCG scores of all users and their corresponding mean of the shortest path profile for the top-$100$ recommended items.
For \lastfm and \delicious we find that the correlations for almost all seeds are negative and for \douban the correlations are positive. For instance, in \lastfm we find that the correlations range from $[-0.31,-0.09]$. For \delicious, the negative correlations range is $[-0.21,-0.006]$ with one positive outlier of $0.17$. In \douban's case, we observe that the users at distances of $3$ are as similar to the target user as the users at distances of $1$ unlike in the other two networks. This allows many users at distances of $2-3$, as suggested by \privaCTCF, to contribute more to the NDCG of the target users thus leading to a positive correlation.

\section{Related Work}
\label{sec:related}

Modelling networks as hierarchical clusters has been widely-studied for more than $70$ years~\cite{florek1951liaison,sibson1973slink}.
The algorithms to mine hierarchical clusters
can be classified into two groups 1) agglomerative/divisive algorithms based
on node similarities, 2) algorithms optimizing cost functions based on
probabilistic modelling of the graph. The first group consists of traditional
algorithms such as single, average and complete linkage, given in the book
Pattern Classification~\cite{duda_hart_stork}. Ward's algorithm
 optimizes for minimum variance of cluster similarities
 ~\cite{ward1963hierarchical}. While these algorithms are simple, they are considered to be ad-hoc and there is no analytical way to compare them. We refer the reader to this survey~\cite{murtagh2012DMKD}.

Unlike the ad-hoc algorithms,
more systematic ones been proposed that optimize certain cost functions. The
examples include, Hierarchical Random Graph model~\cite{clauset2008Nature,clauset2007HRG} and more recently, the Hierarchical Stochastic Block
Model (HSBM)~\cite{lyzinski2015HSBM}. All the aforementioned algorithms are not designed with privacy as a concern and require access to the complete graph for functioning.

Consequently, plenty of works try to address the privacy problem in graphs.
These works~\cite{nissim2007SmoothSensitivity,day2016SIGMOD,cormode2015PGStatsLadder,kasiviswanathan2013NodeDPDegree} release degree distributions,  minimum
spanning trees and subgraph counts under edge/node differential
privacy. More works generated synthetic graphs privately by using the  graph
statistics such as degree counts or edge information as inputs~\cite{wentian14KDD,qian2014KDD,wangyue2013KDD,mir2012EDBT,karwa2012LNCS}. All of them, however, work in the trusted central-aggregator model.

In the LDP setup, many works
have been proposed for frequency estimation and heavy-hitters for categorical
or set-valued data~\cite{bassily2015LDPHist,duchi2023MinimaxRates,qin2016CCS,hsu2012LDPHeavyHitters,rappor2014CCS,fanti2016RapporExt,xiong2016RandomizedRequantization, gu2019LDPRange}. Rappor~\cite{rappor2014CCS} has been used by Google in the past
to collect app statistics. Fanti et al. extends this model to a situation when the
the number of categories is not known beforehand~\cite{fanti2016RapporExt}.
Qin et al. gives a multi round approach like ours to mine heavy hitters from
set-valued data~\cite{qin2016CCS}. The mechanisms are extended to numerical
values and multi-dimensional data in~\cite{wang2019ICDE}. Further, Ye et al.
presents the first frequency and mean estimation algorithms for key-valued
data~\cite{ye2019KeyValueLDP}. Bassily has a more general setup
of estimating linear queries~\cite{bassily2019LDPLinearqueries}. Our work shares the idea
of using minimal tools at our disposal with them.

LDP mechanisms for graphs are less explored. Qin et
al.~\cite{qin2017CCS} proposes a multi round decentralized synthetic graph
generation technique. They use degree vectors to flat cluster users and
then generate edges between users based on their cluster assignments. Our work takes their degree vector insight, however, there are important differences in both problem setup and hence, the techniques used. Our problem setup is different wherein we mine hierarchical cluster trees as opposed to flat clustering of nodes. HCT requires clustering at multiple levels i.e., each node is assigned to many hierarchical clusters, multiple rounds of refining flat clusters cannot be directly applied in our setup. The challenges in learning HCT are unique and require a principle design that builds on some existing insights and many new ones that we talked about in Section~\ref{sec:approach}.

\localalgo differs from other algorithms that have been proposed after the introduction of Dasgupta's cost function. All of them are greedy approaches and some of them require the edge information of the graph. For similarity based clustering, Dasgupta proposed a greedy sparsest-cut based algorithm with $O(cn\log(n))$ approximation. Charikar and Chatziafratis have improved it to a $O(\sqrt{logn})$ approximation by using SDP. Addad et al. summarizes the algorithms and shows that average linkage achieves a $0.5$ approximation dissimilarity setting~\cite{addad2018SIAM}, coinciding with the findings of~\cite{moseley2017NIPS}. In comparison, although slower than greedy,  \localalgo achieves better approximation $1 - O(\frac{\log(n)}{n^2-1})$ in the dissimilarity setting due to its MCMC design. Further, none of the other works have been designed with a privacy objective and it is not straight-forward to argue about their utility under differential privacy (e.g., using degree vectors). In contrast, we propose a privacy-aware and easy to implement algorithm that performs better than the existing greedy approaches. Our design also allows for analyzing its utility guarantees under differential privacy. In fact, our final algorithm turns out to have comparable simplicity to the popular algorithm implemented by open-source libraries, e.g., R packages~\cite{hrgRpackage}.

Social recommender systems are popular and have been studied for over two decades~\cite{basu1998socialrecommendation,tang2013social}. The correlation between social contacts and their influence on user's interests has been observed and theoretically modelled in the real world networks~\cite{golbeck2005computing,singla2008yes,panigrahy2012user,hang2008operators}. When available, the most popular social recommendation algorithms combine both item ratings data as well as social data for collaborative filtering~\cite{benbasat2005trustsocial,jamali2009trustwalker}. The collaborative filtering techniques range from matrix decomposition based~\cite{jiang2012social,guo2015trustsvd,jamali2010matrixsocial,yang2016social,wang2018collaborative} to deep learning based~\cite{fan2019socialgraph}. All of these algorithms however require exact social graph.
Differential private matrix factorization for recommender systems is popular in the centralized setup~\cite{mcsherry2009differentially,zhu2013privaterecommendation}. Recently, this has been studied in the LDP setup as well~\cite{shin2018privaterecommendation,truex2020lprivaterecommendation}. All these techniques have a different privacy setup where in they preserve the privacy of ratings of the user. Further, they do not have any social component associated with them. Instead, in our work we do social recommendation while preserving the privacy of the inherent social network rather than user ratings.
\section{Conclusion \& Future Work}
\label{sec:conclusion}
With millions of users moving towards federated services and social networks, studying how to enable existing analytics applications on these platforms is a new challenge. In this work, we provide the first algorithm to learn a differentially private hierarchical clustering of in a federated network without a trusted authority.
Our approach is principled and follows theoretical analysis which explains why our design choices expect to yield good utility. We apply our new algorithms in social recommendation for the federated setup, replacing privacy-invasive ones, and show promising results. We hope our work encourages future work on supporting richer queries and the full gamut of conventional analytics for federated networks with no trusted co-ordinators. Enabling queries on graphs with private node attributes and on interest graphs are promising next steps.


\appendix
\section{Appendix}
\label{sec:appendix}

\subsection{Proofs}
\label{sec:proofs}

In this section, we give the detailed proofs for the theorems and lemma presented throughout the paper.
\thma*
\begin{proof}

We will prove the claim by contradiction.

Let $\opttree$ maximize $\dasgupta$; assume $\inode_1 \& \inode_2$ be two
internal nodes where $\inode_2 \in \ancestors(\inode_1)$ such that the
$\theta(\inode_1)>\theta(\inode_2)$.

Then, the three configurations of a tree with $\inode_1, \inode_2$ are shown
in the Figure~\ref{fig:localswap}. If configuration $L$ is $\opttree$ then one
of the other two configurations, $M$ or $N$ will have a higher cost than
$\opttree$. Let $|A| = a$, $|B| = b$ and $|C| = c$.

\begin{align}
\text{Given} &,\\
			 &\theta(\inode_2)<\theta(\inode_1)\\
\text{i.e., }&\frac{\sumweights{x}{z}+\sumweights{y}{z}}{(a+b)\cdot c} < \frac{\sumweights{x}{y}}{a\cdot b}\\
\implies     &(a\cdot c\sumweights{x}{y}-a\cdot b\sumweights{x}{z})\\
			 &+(b\cdot c\sumweights{x}{y}-a\cdot b\sumweights{y}{z}) > 0\numberthis\label{eqn1}\\
\end{align}
Next,
\begin{align*}
\dasgupta(M) - \dasgupta(L) &= (a+c)\cdot|\sumweights{x}{z}|\nonumber\\
						    &+ (a+b+c)\cdot|\sumweights{x}{y}+\sumweights{y}{z}|\nonumber\\
						    &- (a+b)\cdot|\sumweights{x}{y}|\nonumber\\
						    &- (a+b+c)\cdot|\sumweights{x}{z}+\sumweights{y}{z}|\nonumber\\
						    &= c\cdot\sumweights{x}{y} - b\cdot\sumweights{x}{z} \numberthis \label{eqn2}\\
\text{Similarly, }\nonumber\\
\dasgupta(N) - \dasgupta(L) &= c\cdot\sumweights{x}{y} - a\cdot\sumweights{y}{z} \numberthis \label{eqn3}\\
\end{align*}

Observe that \eqref{eqn1} implies at least one of \eqref{eqn2} and
\eqref{eqn3} has to be greater than $0$. Hence the contradiction.
\end{proof}

\lemb*
\begin{proof}

Let $\numnodes=|V|$, $\dasgupta(T)$ represent the cost of the sampled tree $T$, from the
stationary distribution and let $\numstates$ be the total number of possible
trees/states. The cost of the optimal tree $\opttree$ is $\optcost$.

\begin{align*}
Pr(\dasgupta(T)\leq m) &\leq \frac{Pr(\dasgupta(T)\leq m)}{Pr(\dasgupta(T)= \optcost)}\\
					   &\leq \frac{\numstates\cdot \frac{e^m}{\sum e^{\dasgupta(T)}}}{1\cdot \frac{e^{\optcost}}{\sum e^{\dasgupta(T)}}} \\
					   &\leq \numstates\cdot e^{m-\optcost}
\end{align*}

Since at maximum only $\numstates$ states can have less  cost than $m$ and assume only one
state has $\optcost$. If we substitute $m = \optcost-\log(\numstates)-t$, then
\begin{align*}
Pr(\dasgupta(T)\leq \optcost - \log(\numstates)-t)&\leq \numstates\cdot e^{-\log(\numstates)-t}\\
												  &\leq e^{-t}\\
Pr(\dasgupta(T)\leq \optcost - c\cdot\numnodes\cdot\log(\numnodes)-t) &\leq e^{-t}\\
Pr(\dasgupta(T)\leq \optcost - c'\cdot\numnodes\cdot\log(\numnodes)) &\leq e^{-\numnodes\cdot\log(\numnodes)}\numberthis \label{eqn4}\\
\end{align*}

Observe that while $\dasgupta$ is a polynomial in $\numnodes$, the probability
of being far from $\optcost$ falls exponentially in $\numnodes$. Therefore, the expected loss in utility is $c'\cdot\numnodes\cdot\log(\numnodes)$ and the loss is heavily concentrated around its expectation.
\end{proof}

\leme*
\begin{proof}
    Let the degree vector of vertex $v_i$ be $\dv{i} = (c_i^1, c_i^2,\cdot{...},c_i^\numbins)$ and the noisy degree vector $\ndv{i} = (\bar{c}_i^1, \bar{c}_i^2,\cdot{...},\bar{c}_i^\numbins)$ such that $\bar{c}_i^l = c_i^1+\nrvi{i}{l}$. Similarly, for vertex $v_j$ $\dv{j} = (c_j^1, c_j^2,\cdot{...},c_j^\numbins)$ and $\ndv{j} = (\bar{c}_j^1, \bar{c}_j^2,\cdot{...},\bar{c}_j^\numbins)$. Recall that the dissimilarity between two vertices is computed as \normL of their degree vector counts
Let's bound the expected difference between original and noisy dissimilarities for two nodes $v_i, v_j$ i.e., $\Expn{R}{|\dpsimscore{v_i}{v_j} - \simscore{v_i}{v_j}|}$.
    \begin{align*}
        &\Expn{R}{\abs{\dpsimscore{v_i}{v_j} - \simscore{v_i}{v_j}}} =\\
        &\Expn{R}{\abs[\Big]{\sum_{l=1}^{l=\numbins}\abs{\ndc{i}{l}-\ndc{j}{l}}
                                                                  -
                                                              \sum_{l=1}^{l=\numbins}\abs{\dc{i}{l}-\dc{j}{l}}}}\\
        \leq&\Expn{R}{\sum_{l=1}^{l=\numbins}\abs{\nrvi{i}{l} - \nrvi{j}{l}}}\\
    \end{align*}
    Note that if $z = \nrvi{i}{l} - \nrvi{j}{l}$, then PDF of z, $f_{\nrvi{i}{l} - \nrvi{j}{l}}(z)$ is
    \begin{align*}
        f_{\nrvi{i}{l} - \nrvi{j}{l}}(z) = \frac{1}{4\cdot\sigma}[e^{\frac{-|z|}{\sigma}} + \frac{|z|}{\sigma}\cdot e^{\frac{-|z|}{\sigma}}]\\
        \therefore \Expn{|z|}{f_{\nrvi{i}{l} - \nrvi{j}{l}}(z)} = \frac{3\sigma}{2}
    \end{align*}

    Finally, we bound the difference between $\dpcost{T}$ and $\actcost{T}$
    \begin{align*}
        \abs[\Big]{\Expn{R}{\dpcost{T} - \actcost{T}}} &=\\ &\abs[\Big]{\sum_{i,j}(\dpsimscore{v_i}{v_j}-\simscore{v_i}{v_j})\cdot p_{i,j}} \\
        &\text{where }p_{i,j}=\abs{\leaves{T[i\lor j]}}\\
        &\leq\sum_{i,j}\abs[\Big]{\Expn{R}{\abs{\dpsimscore{v_i}{v_j}-\simscore{v_i}{v_j}}}\cdot p_{i,j}}\\
        &\leq \frac{3\numbins}{2\epsilon}\cdot\sum_{i,j}(1\cdot\abs{\leaves{T[i\lor j]}})\\
        &\leq \frac{3\numbins}{2\epsilon}\cdot\rho
    \end{align*}
\end{proof}

\lemf*
\begin{proof}
    Let $\numL{i}{j}$ denote $\abs{\leaves{T[i\lor j]}}$. The variance of $\dpcost{T}$ is given by,
    \begin{align*}
        \Var[\dpcost{T}] &= \Var[\sum_{i,j}{\dpsimscore{i}{j}}\cdot\numL{i}{j}]\\
        &=\numL{i}{j}^2\cdot \sum_{i,j}{\Var[\dpsimscore{i}{j}]}\\ &+ 2\sum_{i,j}\sum_{i,k:i\neq j\neq k}\Cov[\dpsimscore{i}{j}, \dpsimscore{i}{k}]\numberthis \label{eqn6}\\
    \end{align*}
    First let's look at $\Var[\dpsimscore{i}{j}]$,
    \begin{align*}
        \Var[\dpsimscore{i}{j}] &= \Var[\sum_{l=1}^{l=\numbins}\abs{\ndc{i}{l} - \ndc{j}{l}}]\\
    \end{align*}
    If $X_l = \abs{\ndc{i}{l} - \ndc{j}{l}}$ then observe that $X_l\forall l\in\{1,2\cdot{...},\numbins\}$ are independent.
    Therefore,
    \begin{align*}
        \Var[\dpsimscore{i}{j}] &= \numbins\cdot\Var[\abs{\ndc{i}{l} - \ndc{j}{l}}]\\ 
        &\text{for some }l\in\{1,2\cdot{...},\numbins\}\\
        &=\numbins\cdot\Var[|\dc{i}{l} - \dc{j}{l}+\nrvi{i}{l} - \nrvi{j}{l}|] \numberthis \label{eqn7}
    \end{align*}
    Fact, $\Var[\abs{X}] = \Var[X]+\Exp[X]^2 - \Exp[\abs{X}]^2$ for any random variable $X$. In our case $X = \abs{\dc{i}{l} - \dc{j}{l}+\nrvi{i}{l} - \nrvi{j}{l}}$ therefore, evaluating each of the terms in R.H.S
    \begin{align*}
        \Var[X] &= \Var[\dc{i}{l} - \dc{j}{l}+\nrvi{i}{l} - \nrvi{j}{l}]\\
                &= \Var[\nrvi{i}{l} - \nrvi{j}{l}]\\
                &= \Var[\nrvi{i}{l}]+\Var[\nrvi{i}{l}]\\
                &= 4\cdot\sigma^2 (\sigma=\frac{1}{\epsilon})\\
        \Exp[X]^2 &= (\dc{i}{l} - \dc{i}{j})^2 (\Exp[\nrvi{i}{l}]=0\forall i\in [n])\\
        \Exp[\abs{X}]^2 &\geq \Exp[X]^2\\
                        &\geq (\dc{i}{l} - \dc{i}{j})^2
        \therefore\\
        \Var[X] &\leq 4\cdot\sigma^2 \numberthis \label{eqn8}
    \end{align*}
    From Equations~\eqref{eqn7},\eqref{eqn8},
    \begin{align*}
        \Var[\dpsimscore{i}{j}]\leq 4\numbins\sigma^2 \numberthis \label{eqn20}
    \end{align*}
    Now let's bound the covariance term from Equation~\ref{eqn6} i.e., $\Cov[\dpsimscore{i}{j}, \dpsimscore{i}{k}]$. Similar to $X_l$ assume $Y_l = \abs{\ndc{i}{l} - \ndc{k}{l}}$
    \begin{align*}
        \dpsimscore{i}{j} &= \abs[\Big]{\sum_{l=1}^{l=\numbins}X_l}\\
        \dpsimscore{i}{k} &= \abs[\Big]{\sum_{l=1}^{l=\numbins}Y_l}
    \end{align*}
    Observe that $X_i, Y_j$ when $i\neq j$ are independent. Therefore,
    \begin{align*}
        \Cov[\dpsimscore{i}{j}, \dpsimscore{i}{k}] &= \\
        &\sum_{l=1}^{l=\numbins}\Cov[X_l, Y_l]\numberthis \label{eqn9}\\
    \end{align*}
    We will investigate $\Cov[X_l, Y_l]$ for some $l$,
    \begin{align*}
        \Cov[X_l, Y_l]&=\Cov[\abs{\ndc{i}{l} - \ndc{j}{l}}, \abs{\ndc{i}{l} - \ndc{k}{l}}]\\
        &=\Cov\left[\abs{\dc{i}{l}-\dc{j}{l}+\nrvi{i}{l}-\nrvi{j}{l}}\right.,\\
        & \left.\abs{\dc{i}{l}-\dc{k}{l}+\nrvi{i}{l}-\nrvi{k}{l}}\right]
    \end{align*}
    Observe that $\nrvi{i}{l}-\nrvi{j}{l}$, $\nrvi{i}{l}-\nrvi{k}{l}$ are identical random variables but not independent, therefore, we replace both of them with a random variable $X_{id}$.
    \begin{align*}
        \Cov[X_l, Y_l]&=\left[\abs{\dc{i}{l}-\dc{j}{l}+X_{id}},\right.\\
        & \left.\abs{\dc{i}{l}-\dc{k}{l}+X_{id}}\right]\\
        &=\left[\abs{a+X_{id}}, \abs{b+X_{id}}\right]\\
        &\text{replacing }a+X_{id}\text{ with }X'_{id}\numberthis \label{eqn10}\\
        &=\left[\abs{X'_{id}}, \abs{(b-a)+X'_{id}}\right]\\
    \end{align*}
    Fact, $\Cov\left[\abs{X}, \abs{a+X}\right]\leq \Cov\left[\abs{X}, \abs{X}\right]$.
    \begin{align*}
        \therefore,
        \Cov[X_l, Y_l]&\leq\left[\abs{X'_{id}}, \abs{X'_{id}}\right]\\
        &=\Var[\abs{X'_{id}}]\\
        &=\Var\left[\abs{a+X_{id}}\right]\text{ From~\ref{eqn10}}\\
        &=\Var\left[\abs{a+\nrvi{i}{l}-\nrvi{j}{l}}\right]\\
        &\leq 4\cdot\sigma^2\numberthis \label{eqn11}\text{ From Equation~\ref{eqn8}}\\
    \end{align*}
    Similar analysis works for all $l$, therefore, $\Cov[X_l, Y_l]\leq4\cdot\sigma^2 \forall l\in\{1,2,\cdot{...},\numbins\}$. Therefore, from Equation~\ref{eqn9},\ref{eqn11}
    \begin{align*}
         \Cov[\dpsimscore{i}{j}, \dpsimscore{i}{k} \leq 4\numbins\sigma^2 \numberthis \label{eqn12}
    \end{align*}
    From Equations~\ref{eqn6},\ref{eqn20},\ref{eqn12} this follows
    \begin{align*}
        \Var[\dpcost{T}]&\leq 4\numbins\sigma^2\cdot\left(\sum_{i,j}\numL{i}{j}^2+2\cdot\sum_{i,j}\sum_{i,k}\numL{i}{j}\numL{i}{k}\right)\\
        &\leq 4\numbins\sigma^2\cdot\left(\sum_{i,j}1\cdot\numL{i}{j}\right)^2\\
        &\leq 4\numbins\sigma^2\cdot\rho^2
    \end{align*}
    Hence Proved.
\end{proof}

\lemg*
\begin{proof}
    We know from Theorem~\ref{thm:t6},
    \begin{align*}
        \abs[\Big]{\Expn{R}{\dpcost{T} - \actcost{T}}}\leq \frac{3\numbins}{2\epsilon}\cdot\rho
    \end{align*}
    Also, we know a bound on variance of $\dpcost{T}$ from theorem~\ref{thm:t7}.
    \begin{align*}
         \Var[\dpcost{T}]&\leq 4\numbins\sigma^2\cdot\rho^2
    \end{align*}
    
    Applying Chebyshev's inequality,
    \begin{align*}
        Pr\left(\abs[\Big]{\dpcost{T}-\Expn{R}{\dpcost{T}}}\geq \delta\right)&\leq \frac{\Var[\dpcost{T}]}{\delta^2}\\
        \text{Substituting,   }\delta&=\frac{6\sqrt{\numbins}}{\epsilon}\cdot\rho\\
        Pr\left(\abs[\Big]{\dpcost{T}-\Expn{R}{\dpcost{T}}}\geq \frac{6\sqrt{\numbins}}{\epsilon}\cdot\rho\right) &\leq \frac{1}{9}\\
        \therefore \abs{\dpcost{T}-\actcost{T}}\leq\frac{3\numbins}{2\epsilon}\cdot\rho+\frac{6\sqrt{\numbins}}{\epsilon}\cdot\rho\text{ w.h.p }\numberthis\label{eqn21}
    \end{align*}
    Hence proved.
\end{proof}

\thmh*
\begin{proof}
    In the \dplocalalgo, the sampling is done based on the noisy probabilities $Pr(T_{dp}) = \frac{e^{\dpcost{T_{dp}}}}{\sum_{T'}e^{\dpcost{T'}}}$. In order to measure utility loss, the cost of sampled tree $T_{dp}$ is computed using $\actcost{T_{dp}}$. Let $T$ be a tree with $f(T)=m$ and the set of trees with cost less than $\actcost{T}$ are $T_{set}$ and $T''$ has the highest probability to be sampled among them(Note that, if we sampled using $\actcost{T_{dp}}$ then $T_{dp}$ would have the least cost and highest probability). If
    \begin{align*}
        P=\frac{2\numbins}{\epsilon}\cdot\left(\frac{3}{2}+\frac{6}{\sqrt{\numbins}}\right)\cdot\rho
    \end{align*}
    then,
    \begin{align*}
        Pr(\actcost{T_{dp}}\leq m)&\leq\frac{Pr(\actcost{T_{dp}}\leq m)}{Pr(\actcost{T_{dp}}=\optcost)}\\
        &\leq \abs{T_{set}}\cdot e^{\dpcost{T''}-\dpcost{\opttree}}\\
        &\leq \numstates\cdot e^{(\actcost{T''}+P) -(\actcost{\opttree} - P)} \\
        &\text{ w.h.p, From lemma~\ref{thm:t7}}\\
        &\leq \numstates\cdot e^{2P}\cdot e^{\actcost{T}-\optcost} \text{ replace T'' with T }\\
    \end{align*}
    Now, given lemma~\ref{thm:t7} w.h.p, substitute $m=\optcost-2P-2\cdot\log{\numstates}$
    \begin{align*}
        Pr(\actcost{T_{dp}}\leq \optcost-2P-2\cdot\log{\numstates})&\leq e^{-\numnodes\log{\numnodes}}\\
        \text{ since, }\log{\numstates}=c \cdot \numnodes\log{\numnodes}\numberthis\label{eqn13}\\
        \therefore \text{Expected Utility Loss } &= 2P\\
        &=\frac{2\numbins}{\epsilon}\cdot\left(\frac{3}{2}+\frac{6}{\sqrt{\numbins}}\right)\cdot\rho\\
        \text{ the }\log{\numstates}\text{ is }<< P\\
    \end{align*}
   Hence proved.
\end{proof}

\subsection{Evaluation Metrics for Social Recommendations}
\label{sec:eval-metrics}

We use mean average precision (MAP) and normalized
discounted cumulative gain (NDCG) for evaluating the performance of our recommendation algorithms. These two metrics give a sense of how relevant the top-$k$ recommendations are based on their presence as well as the order in which they are presented to the user. MAP is the mean of average precisions that are computed at each cut off point in a sequence of recommendations. The precision $p_i$ at the cut-off point of $i$ is defined as the number of true positives in the sequence of $i$ recommendations divided by $i$.
The relevance $r_i$ of an item $i$ is an indicator function which equals $1$ if
the item at rank $i$ is relevant and $0$ otherwise. 
The average precision $AP_k$ and the mean average precision $MAP_k$ are thus:
\begin{align*}
  AP_k = \frac{1}{N}\sum_{i}^{k} p_i \cdot r_i\text{;    }MAP_k = \frac{1}{k}{\sum_{i=1}^{k} AP_i}
\end{align*}

The normalized discounted gain is the discounted cumulative gain (DCG) normalized by
the ideal DCG (IDCG) score.
$DCG_k$ is a metric that evaluates the gain of an item based on its position in
the top-$k$ recommendation. In $DCG_k$, we penalize when highly relevant items
appear lower in the recommendation list by reducing their relevance $r_i$ is
reduced by the position $i$ in the top-$k$ list.
$IDCG_k$ represents the DCG for top-$k$ sorted items in the ground truth. 
We consider the relevance as binary with $r_i=1$ if the predected top-$K$ item is there in the ground truth.
\begin{align*}
  DCG_k = \sum_{i=1}^k\frac{r_i}{\log_2{(i+1)}}\text{;    }NDCG_k = \frac{DCG_k}{IDCG_k}
\end{align*}

\subsection{LDPGen non-reproducibility}
\label{sec:ldpgen}

We detail our efforts to reproduce both experiments and theoretical analysis of prior work that is closely related to ours~\cite{qin2017CCS}. The paper provides a good insight with regards to using degree vectors to capture edge information, however, we find some issues with the paper which might be helpful to the broader community if we point them out.

\paragraph{The formula for number of clusters is not valid for their datasets} The paper mentions in Section 4.2 (page 431 in the proceedings) {\em ``... Fortunately, by limiting $d_{x,y}$ and $k_1$ to a range (e.g., between $0$ and $50$, which is suitable for most social networks), we get a good approximated closed form as follows ...''}.
However, the networks which have been evaluated have a significant number of nodes with degree more than $50$. 
So, the closed-form expression (20) for the number of clusters $k_1$ does not apply to the evaluated datasets. In our paper, we  heuristically choose to use $K=\log{n}$ which applies regardless of the degree.

\paragraph{The formula for probability of edge between two nodes u,v does not seem to be correct}
The formula given in Section 4.4 is not symmetric, i.e., $p_{u,v}$ is not equal to $p_{v,u}$. Therefore, for undirected graphs it is not clear how to choose the probability. Further, the probability formula does not correspond to the Chung-Lu model as mentioned there. Therefore, we tried to replace the formula with our own expertise, however, we could not reproduce their results using our formulas.

\paragraph{We cannot reproduce the recommender system performance.}
LDPGen evaluates recommender system performance on the \lastfm and flixster datasets (Section 5.2, page 435). To the best of our ability, we have looked for state-of-the-art social recommendation systems that are not private. We cannot reproduce this for the reasons we stated above. Moreover, the state-of-the-art non-private baselines seem to record an NDCG score of $0.3$ on \lastfm~\cite{wang2018collaborative}. However, LDPGen based differentially private technique records an NDCG score of $0.8$.

\paragraph{The epsilon values reported are half of what they should be.} In the LDP setup, the same edge is queried twice, once by each user. Therefore, the privacy loss is double that of the one used by each user to answer the query.

One of the authors have confirmed the problems we have mentioned, however, the code was not given to us.  Therefore, after spending considerable time to figure out the right formulas and metrics used as per our interpretation we consider it non-reproducible.

\subsection{Observations that help \privaCTCF}
\label{sec:appobs}

The correlations between user-user similarities and their shortest path distances on the network are given in Table~\ref{tab:corr}.

\begin{table}[t]
\caption{ The correlations between the user-user similarities and the shortest path distances between them is given here. The negative values signify that most users are similar to other closer users in the network than the farther ones based on the shortest path distances between them.
}
\begin{tabular}{clcc}
                            &  \textbf{Rating profile}    & \textbf{\% negative corr.} & \textbf{corr. range}\\ \hline
\multirow{2}{*}{\lastfm}     & Partial & 76.8                   & [-0.25, 0.17]                 \\ \cline{2-4} 
                            & Full  & 100.0                  & [-0.23, 0.02]                 \\ \hline
\multirow{2}{*}{\delicious} & Partial & 96.6                   & [-0.26, 0.04]                 \\ \cline{2-4} 
                            & Full  & 100.0                   & [-0.23, 0.02]                  \\ \hline
\multirow{2}{*}{\douban}    & Partial & 95.1                   & [-0.14, 0.05]                  \\ \cline{2-4}
                            & Full  & 100.0                   & [-0.28, -0.01]                 \\ \hline
\end{tabular}
\label{tab:corr}
\end{table}

\end{document}